\documentclass[a4paper]{article}
\usepackage{fullpage}
\usepackage[T1]{fontenc}
\usepackage[utf8]{inputenc}
\usepackage[all=normal,bibliography=tight]{savetrees}

\usepackage{amstext,amsfonts,amsthm,amsmath,amssymb}

\usepackage{graphicx,tikz}
\usetikzlibrary{snakes,shapes}
\usetikzlibrary{calc}
\usetikzlibrary{decorations.markings}
\usetikzlibrary{patterns}
\usepackage{comment}
\usepackage{subcaption}
\usepackage{url}
\usepackage{xspace}
\usepackage[textsize=footnotesize,backgroundcolor=white]{todonotes}
\usepackage[shortlabels]{enumitem}
\setlist{topsep=1ex,itemsep=-1ex,partopsep=0ex,parsep=1ex}
\usepackage[absolute]{textpos}
\usepackage[colorlinks=true, citecolor=red]{hyperref}
\usepackage[capitalize]{cleveref}
\usepackage{wrapfig}
\usepackage{rotating}
\usepackage{pdflscape}
\usepackage{float}
\usepackage[noadjust]{cite} 

\def\cqedsymbol{\ifmmode$\lrcorner$\else{\unskip\nobreak\hfil
\penalty50\hskip1em\null\nobreak\hfil$\lrcorner$
\parfillskip=0pt\finalhyphendemerits=0\endgraf}\fi}

\newcommand{\Oh}{\mathcal{O}}

\newcommand{\X}{\mathcal{X}}

\DeclareMathOperator{\tw}{\mathbb{tw}}
\DeclareMathOperator{\Red}{Red}
\DeclareMathOperator{\pred}{pred}
\newcommand{\poly}{poly}
\newcommand{\degg}{d}

\newcommand{\defproblem}[4]{
 \vspace{5mm}
\noindent\fbox{
 \begin{minipage}{0.85\textwidth}
 \begin{tabular*}{\textwidth}{@{\extracolsep{\fill}}lr} #1 & \\ \end{tabular*}
 {\bf{Input:}} #2 \\
 {\bf{Parameter:}} #3 \\
 {\bf{Question:}} #4
 \end{minipage}
 }
 \vspace{5mm}
}

\newtheorem{lemma}{Lemma}[section]

\newtheorem{observation}[lemma]{Observation}

\newtheorem{theorem}[lemma]{Theorem}
\newtheorem{claim}[lemma]{Claim}
\theoremstyle{definition}
\newtheorem{definition}[lemma]{Definition}
\newtheorem*{reduction*}{Reduction Rule}

\crefname{lemma}{Lemma}{Lemmas}
\crefname{theorem}{Theorem}{Theorems}

\title{Close relatives (of Feedback Vertex Set), revisited\thanks{This research is part of a project
that has received funding from the European Research Council (ERC)
under the European Union's Horizon 2020 research and innovation programme
Grant Agreement 714704. This research was conducted while Hugo Jacob was doing a research internship
at the University of Warsaw.}}

\author{
         Hugo Jacob\footnote{ENS Paris-Saclay, France, \texttt{hugo.jacob@ens-paris-saclay.fr}}
    \and Thomas Bellitto\footnote{Faculty of Mathematics, Informatics and Mechanics, University of Warsaw, Poland, \texttt{tbellitto@mimuw.edu.pl}}
    \and Oscar Defrain\footnote{Faculty of Mathematics, Informatics and Mechanics, University of Warsaw, Poland, \texttt{odefrain@mimuw.edu.pl}}
    \and Marcin Pilipczuk\footnote{Faculty of Mathematics, Informatics and Mechanics, University of Warsaw, Poland, \texttt{malcin@mimuw.edu.pl}}
}

\date{}

\begin{document}

\maketitle

\begin{textblock}{20}(0, 12.5)
\includegraphics[width=40px]{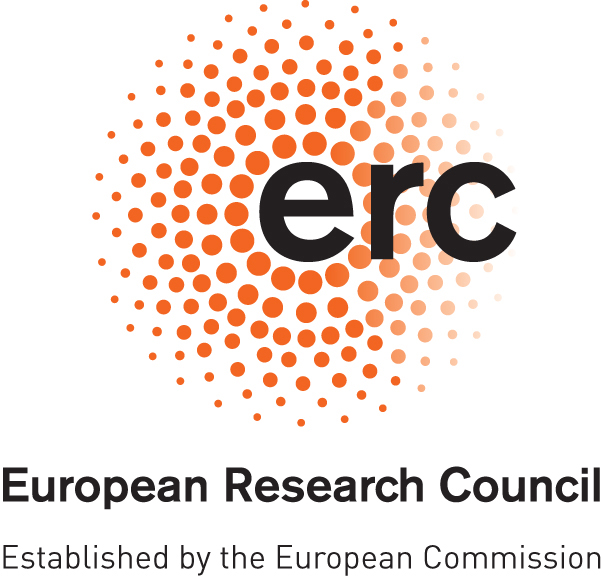}%
\end{textblock}
\begin{textblock}{20}(0, 13.4)
\includegraphics[width=40px]{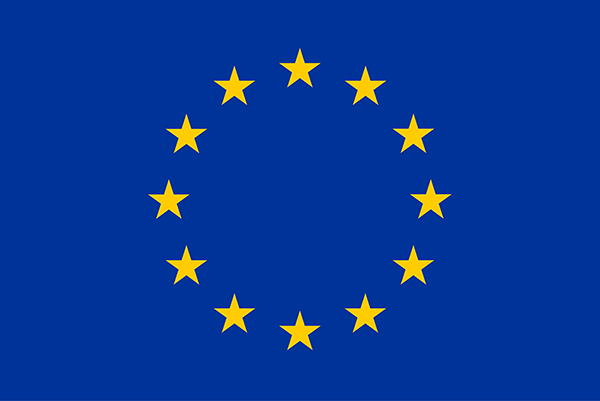}%
\end{textblock}
\begin{abstract}
At IPEC 2020, Bergougnoux, Bonnet, Brettell, and Kwon (\textit{Close Relatives of Feedback Vertex Set Without Single-Exponential Algorithms Parameterized by Treewidth}, IPEC 2020, LIPIcs vol.~180, pp.~3:1--3:17) 
showed that a number of problems related to the classic \textsc{Feedback Vertex Set} (\textsc{FVS})
  problem 
do not admit a $2^{o(k \log k)} \cdot n^{\Oh(1)}$-time algorithm on graphs of treewidth at most $k$, assuming the Exponential Time Hypothesis. 
This contrasts with the $3^{k} \cdot k^{\Oh(1)} \cdot n$-time algorithm for \textsc{FVS} using the Cut\&Count technique. 

During their live talk at IPEC 2020, Bergougnoux et al.~posed a number of open questions, which we answer in this work.
\begin{itemize}
\item \textsc{Subset Even Cycle Transversal}, \textsc{Subset Odd Cycle Transversal}, \textsc{Subset Feedback Vertex Set} can be solved in time $2^{\Oh(k \log k)} \cdot n$ in graphs of treewidth at most $k$. This matches a lower bound for \textsc{Even Cycle Transversal} 
of Bergougnoux et al.~and improves the polynomial factor in some of their upper bounds. 
\item \textsc{Subset Feedback Vertex Set} and \textsc{Node Multiway Cut} can be solved in 
time $2^{\Oh(k \log k)} \cdot n$, if the input graph is given as a cliquewidth expression of size $n$ and width $k$.
\item \textsc{Odd Cycle Transversal} can be solved in time $4^k \cdot k^{\Oh(1)} \cdot n$
if the input graph is given as a cliquewidth expression of size $n$ and width $k$.
Furthermore, the existence of a constant $\varepsilon > 0$ and an algorithm performing this
task in time $(4-\varepsilon)^k \cdot n^{\Oh(1)}$ would contradict the Strong Exponential Time Hypothesis. 
\end{itemize}
A common theme of the first two algorithmic results is to represent connectivity properties of the current graph in a state of a dynamic programming algorithm as an auxiliary forest with $\Oh(k)$ nodes.
This results in a $2^{\Oh(k \log k)}$ bound on the number of states for one node of the tree decomposition or cliquewidth expression and allows to compare two states in $k^{\Oh(1)}$ time, resulting in linear time dependency on the size of the graph or the input cliquewidth expression.

\end{abstract}

\section{Introduction}
Treewidth, introduced by Robertson and Seymour in their seminal Graph Minors series~\cite{RobertsonS84}, but also independently introduced under different names by other authors, is probably
the most successful graph width notion. 
(For the formal definition of treewidth and other width notions mentioned in this introduction, we refer to Section~\ref{sec:prelims}.)
From the algorithmic point of view, its applicability is described by Courcelle's theorem~\cite{Courcelle90}
that asserts that every problem expressible in monadic second order logic with quantification
over vertex sets and edge sets, can be solved in linear time on graphs of bounded treewidth.

Due to the abundance of algorithms for graphs of bounded treewidth applications
and since Courcelle's theorem provides a very weak bound on the dependency of the running time 
of the algorithm on the treewidth of the input graph, 
a lot of research in the last decade has been devoted to understanding optimal running time 
bounds for algorithms on graphs of bounded treewidth. 
One of the first methodological approaches was provided by two works of Lokshtanov, Marx,
and Saurabh from SODA 2011~\cite{lokshtanov2011known,LokshtanovMS18known,LokshtanovMS11soda,LokshtanovMS18slightly}. Their contribution can be summarized as follows.
\begin{itemize}
\item For a number of classic problems,
the known (and very natural) dynamic programming algorithm,
given an $n$-vertex graph $G$ and a tree decomposition of width $k$,
runs in time $c^k \cdot n^{\Oh(1)}$ for a constant $c > 1$. \cite{lokshtanov2011known,LokshtanovMS18known} shows that in most cases the constant $c$
is optimal, assuming the Strong Exponential Time Hypothesis.%
\footnote{For a discussion on the complexity assumptions used, namely the Exponential Time Hypothesis (ETH) and the Strong Exponential Time Hypothesis (SETH), we refer to Chapter 14 of~\cite{cygan2015parameterized}.}
\item \cite{LokshtanovMS11soda,LokshtanovMS18slightly} introduces a framework for proving lower bounds (assuming the Exponential Time Hypothesis) against
$2^{o(k \log k)}\cdot n^{\Oh(1)}$-time algorithms with the same input as above.
\end{itemize}
Both aforementioned works seemed to point to a general conclusion that the natural and naive dynamic programming
algorithms on graphs of bounded treewidth are probably optimal in essentially all interesting cases. 
This intuition has been refuted by Cygan et al.~\cite{CyganNPPRW11} who presented the Cut\&Count technique at FOCS 2011 which allowed $2^{\Oh(k)}\cdot n^{\Oh(1)}$-time algorithms on graphs of treewidth $k$
for many connectivity problems where the natural and naive algorithm runs in time
$2^{\Oh(k \log k)}\cdot n^{\Oh(1)}$. 
One of the prominent examples of such problems is \textsc{Feedback Vertex Set} (\textsc{FVS}) where, given a graph $G$ and an integer $p$, one asks for a set of at most $p$ vertices
that hits all cycles of $G$.

Since then, the intricate landscape of optimal algorithms parameterized by the treewidth
has been explored by many authors, see 
e.g.~\cite{%
DBLP:conf/mfcs/Pilipczuk11,
DBLP:journals/iandc/BodlaenderCKN15,
DBLP:journals/iandc/CyganMPP17,
DBLP:conf/wg/BonamyKNPSW18,
DBLP:journals/jacm/CyganKN18,
DBLP:journals/corr/abs-1907-04442,
DBLP:conf/mfcs/SauS20,
DBLP:journals/siamdm/BasteST20,DBLP:journals/tcs/BasteST20,DBLP:journals/jcss/BasteST20}.
Last year at IPEC 2020, Bergougnoux, Bonnet, Brettell, and Kwon~\cite{bergougnoux2020close}
presented an in-depth study of problems related to \textsc{FVS}, showing that for most of
them $2^{\Oh(k \log k)}$ is the optimal (assuming ETH) dependency on treewidth
in the running time bound. 
During their live talk at IPEC 2020, they asked a number of open questions. In this work, we continue this line of research and answer all of them.

\paragraph*{Hitting cycles in graphs of bounded treewidth.} We first focus on the problems
\textsc{Odd Cycle Transversal} (\textsc{OCT}) and \textsc{Even Cycle Transversal} (\textsc{ECT})
where, given a graph $G$ and an integer $p$, the goal is to pick a set of at most $p$ vertices of
$G$ that hits all odd cycles (resp.~even cycles) of $G$.
These problems are thus closely related to the aforementioned \textsc{FVS} 
problem that asks to hit \emph{all} cycles.
Using the fact that graphs without odd cycles are exactly bipartite graphs, 
it is relatively easy to obtain a $3^k\cdot k^{\Oh(1)}\cdot n$-time algorithm for \textsc{OCT} for
graphs equipped with a tree decomposition of width $k$~\cite{FioriniHRV08}, and the base $3$ of the
exponent is optimal assuming SETH~\cite{lokshtanov2011known,LokshtanovMS18known}. 

In contrast to \textsc{FVS} and \textsc{OCT}, Bergougnoux et al.~\cite{bergougnoux2020close}
showed that, assuming ETH, \textsc{ECT} admits no $2^{o(k \log k)}\cdot n^{\Oh(1)}$-time algorithm
and asked for a matching upper bound.
Our first result is a positive answer to this question, even in a more general setting
of \textsc{Subset Odd Cycle Transversal} (SOCT) and \textsc{Subset Even Cycle Transversal} (SECT)
where $G$ is additionally equipped with a set $S \subseteq V(G)$, and we are only required to hit odd (resp.~even) cycles that pass through at least one vertex of $S$.
\begin{theorem}\label{thm:subset-ect}
\textsc{Subset Odd Cycle Transversal} and \textsc{Subset Even Cycle Transversal},
  even in the weighted setting,
  can be solved in time $2^{\Oh(k \log k)} \cdot n$
on $n$-vertex graphs of treewidth $k$.
\end{theorem}
Here, and in later statements, by \emph{weighted setting} we mean the following: every vertex has its positive integer weight, and the input integer $p$ becomes an upper bound on the total weight of the solution.

Misra, Raman, Ramanujan, and Saurabh~\cite{misra2012parameterized} showed that a graph $G$
does not contain an even cycle if and only if every block (2-connected component) of $G$ is an edge or an odd cycle. 
The key ingredient of the proof of Theorem~\ref{thm:subset-ect} for \textsc{SECT}
is a characterization
(in the same spirit, but more involved) of graphs $G$ with sets $S \subseteq V(G)$
that do not contain an even cycle passing through a vertex of~$S$. 

In \textsc{Subset Feedback Vertex Set} (\textsc{SFVS}), given a graph $G$, a set $S \subseteq V(G)$, and an integer $p$,
the goal is to find a set of at most $p$ vertices that hits every cycle that passes through a vertex of $S$. 
We remark that, as it is straightforward to reduce SFVS
to SECT without increasing the treewidth (just subdivide
every edge once), the running time bound of Theorem~\ref{thm:subset-ect} applies
also to \textsc{SFVS}. This improves the polynomial factor of the running time bound
of~\cite{bergougnoux2020close} for \textsc{SOCT} and \textsc{SFVS} from cubic to linear. 

\paragraph*{Clique-width parameterization.}
We then switch our attention to clique-width. 
Clique-width is a width measure aiming at capturing simple yet (contrary to treewidth) dense graphs.
It originates from works of Courcelle, Engelfriet, and Rozenberg~\cite{CourcelleER93}
and of Wanke~\cite{Wanke94} from early 90s. 
Informally speaking, a graph $G$ is of clique-width at most $k$ if one can provide an expression
(called a \emph{$k$-expression})
that constructs $G$ using only $k$ \emph{labels} which essentially are names for vertex sets. 
Clique-width plays the role of treewidth for dense graphs in the following sense:
any problem expressible in monadic second order logic with quantification over vertex sets 
(but not edge sets) can be solved in time $f(k) \cdot n$, given a $k$-expression of size
$n$ constructing the input graph, where $f$ is some computable function~\cite{CourcelleMR00}. 
Similarly as for treewidth, it is natural to investigate optimal functions $f$ in such running time bounds. 
Here, the most relevant work is due to Bui-Xuan, Such\'{y}, Telle, and Vatshelle~\cite{Bui-XuanSTV13} that showed an algorithm with $f(k) = 2^{\Oh(k \log k)}$ for \textsc{FVS}. 

One should also mention a long line of work~\cite{FominGLS10,FominGLS14,FominGLSZ19} searching for optimal running time bounds on graphs of bounded clique-width for problems \emph{not} captured by the aforementioned meta-theorem and that provably (unless $\mathrm{FPT} = \mathrm{W}[1]$) do not have algorithms with the running time bound $f(k) \cdot n^{\Oh(1)}$, given a $k$-expression building the input graph.

Following on the open questions provided by Bergougnoux et al., we focus
on \textsc{SFVS} and \textsc{Node Multiway Cut} (\textsc{NMwC}). 
In the second problem, given a graph $G$, a set $T \subseteq V(G)$, and an integer $p$,
the goal is to find a set of at most $p$ vertices that does not contain any vertex
of $T$, but hits all paths with both endpoints in $T$.   
We show the following.
\begin{theorem}\label{thm:sfvs}
\textsc{Subset Feedback Vertex Set} and \textsc{Node Multiway Cut}, even in the weighted setting,
 can be solved
in time $2^{\Oh(k \log k)} \cdot n$
if the input graph is given as a $k$-expression of size $n$.
\end{theorem}
Note that the running time bound of Theorem~\ref{thm:sfvs} matches the lower bound
of Bergougnoux et al.~\cite{bergougnoux2020close} for pathwidth parameterization\footnote{We do not formally define pathwidth in this work, as it is not used except for this paragraph.}
of \textsc{SFVS} and \textsc{NMwC}, and it is straightforward to turn a path decomposition
of width $\ell$ into a $k$-expression for $k = \ell + \Oh(1)$. 

Observe also that, if vertex weights are allowed, \textsc{NMwC} reduces to \textsc{SFVS}.
Namely, given a \textsc{NMwC} instance $(G,T,p)$, set the weights of all vertices of $T$
to $+\infty$, create a graph $G'$ by adding to $G$ a new vertex $s$ of weight $+\infty$
adjacent to all vertices of $T$ and set $S := \{s\}$; 
the \textsc{SFVS} instance $(G',S,p)$ is easily seen to be equivalent to the input \textsc{NMwC} instance $(G,T,P)$.
Since it is straightforward to turn a $k$-expression of $G$ into a $(2k)$-expression of $G'$,
in Theorem~\ref{thm:sfvs} it suffices to focus only on the \textsc{SFVS} problem.

A common theme in the dynamic programming algorithm of Theorem~\ref{thm:subset-ect}
and of Theorem~\ref{thm:sfvs}
is the representation of the connectivity in the currently analyzed graph as an auxiliary 
forest of size $\Oh(k)$ with some annotations.
This allows a neat description of the essential connectivity features, 
avoiding involved case analysis. 
The $\Oh(k)$ bound serves two purposes. First, it implies a bound of $2^{\Oh(k \log k)}$
on the number of states of the dynamic programming algorithm
at one node of the tree decomposition or $k$-expression.
Second, it allows to perform computations on states in $k^{\Oh(1)}$ time, giving
the final linear dependency on the size of the graph or the input $k$-expression
in the running time bound.

\paragraph*{Hitting odd cycles in graphs of bounded clique-width.}
Finally, we restrict our attention to \textsc{Odd Cycle Transversal}. Recall that 
in graphs of treewidth $k$, \textsc{OCT} admits an algorithm with running time bound
$3^k \cdot k^{\Oh(1)} \cdot n$~\cite{FioriniHRV08} and the base $3$ is optimal assuming
SETH~\cite{lokshtanov2011known,LokshtanovMS18known}. 
We show that for clique-width, the optimal base is $4$.
\begin{theorem}\label{thm:oct}
\textsc{Odd Cycle Transversal}, even in the weighted setting, can be solved in time $4^k \cdot k^{\Oh(1)} \cdot n$
if the input graph is given as a $k$-expression of size $n$.
Furthermore, the existence of a constant $\varepsilon > 0$ and an algorithm performing the same task
in time $(4-\varepsilon)^k \cdot n^{\Oh(1)}$ contradicts the Strong Exponential Time Hypothesis. 
\end{theorem}
The key insight in the \textsc{OCT} algorithm of~\cite{FioriniHRV08} is to reformulate the problem
into finding explicitly a partition $V(G) = X \uplus A \uplus B$ that minimizes $|X|$ while keeping $G[A]$ and $G[B]$ both edgeless. 
Then, in a dynamic programming algorithm on a tree decomposition,
one remembers the assignment of the vertices of the current bag into $X$, $A$, and $B$; this yields the $3^k$ factor in the time complexity.
For clique-width, a similar approach yields $4^k$ states: every label may be allowed to contain only vertices of $X$, allowed to contain vertices of $X$ or $A$ but not $B$, allowed to contain vertices of $X$ or $B$ but not $A$, or allowed to contain vertices of any of the three sets. 
To obtain the upper bound of Theorem~\ref{thm:oct}, one needs to add on top of the above an appropriate convolution-like treatment of the disjoint union nodes of the $k$-expression. 
The lower bound of Theorem~\ref{thm:oct} combines a way to encode evaluation of two variables of a CNF-SAT formula into one of the four aforementioned states of a single label with 
a few gadgets for checking in the \textsc{OCT} regime if a clause is satisfied, borrowed from the corresponding reduction for pathwidth from~\cite{lokshtanov2011known,LokshtanovMS18known}.

\medskip
After preliminaries in Section~\ref{sec:prelims}, Theorem~\ref{thm:subset-ect} is proven
in Section~\ref{sec:ect} and Theorem~\ref{thm:sfvs} is proven in Section~\ref{sec:sfvs}.
Finally, Theorem~\ref{thm:oct} is proven in Section~\ref{sec:oct}.

\section{Preliminaries}\label{sec:prelims}
\def\tw{\mathbf{tw}}
\def\cw{\mathbf{cw}}
\def\lcw{\mathbf{lcw}}

In this paper, the notation $[n]$ stands for $\{1,\dots,n\}$.
A \emph{$k$-coloring} or \emph{$k$-labeling} of a graph is a mapping $\alpha$ from its vertices to $[k]$.
A coloring is said to be proper if for every two adjacent vertices $u$ and $v$, $\alpha(u)\neq \alpha(v)$.

A \emph{multigraph} is a graph where there can be several edges between a given pair of vertices. 
In this paper, we do not consider graphs with loops, i.e., edges whose two extremities are the same vertex.
If every pair of vertices is connected by at most one edge, we say that the graph is simple.
For a multigraph graph $G=(V,E)$ and a set of edges $A$, we denote by $G+A$ the graph whose vertex set is $V$ and whose edges are all the edges of $E$ and all the edges of $A$. 
Note that this may create multiple edges even if $G$ is a simple graph. 
For a subset $S$ of vertices of $V$, the subgraph of $G$ induced by $S$ is denoted by $G[S]$ and is the graph whose vertex set is $S$ and whose edges are all the edges of $E$ whose both extremities are in $S$. 
We denote by $G- S$ the graph $G[V\setminus S]$.
For two graphs $G$ and $H$, their union $G\cup H$ is the graph whose vertex set is the union of the vertex sets of $G$ and $H$ and that has all the edges of both $G$ and $H$. Again, if two vertices $u$ and $v$ belong to both $G$ and $H$ and are adjacent in both, $G\cup H$ can have multiple edges even if $G$ and $H$ are simple.
For a graph $G$ and two vertices $u$ and $v$, identifying $u$ and $v$ is the operation that replaces the vertices $u$ and $v$ by a new vertex $w$ and replace every edge $xu$ and every edge $xv$ by an edge $xw$.
Again, if $u$ and $v$ shared common neighbours, this operation creates multiple edges. 
For an edge $uv$ of a graph $G$, \emph{subdiving $uv$} is the operation that replaces the edge $uv$ with a new vertex $w$ and edges $uw$ and $wv$. 
For a graph $G$, a \emph{clique} of $G$ is a set of pairwise adjacent vertices and an \emph{independent set} is a set of pairwise nonadjacent vertices.
A \emph{biclique} is a complete bipartite graph.

A \emph{rooted tree} is a tree $T$ together with a special vertex $r\in V(T)$, called the \emph{root}.
It induces a natural ancestor-descendent relation $\leq$ on its vertex set, where a vertex $s\in V(T)$ is said to be a \emph{descendent} of a vertex $t\in V(T)$, denoted $s\leq t$, if $t$ is on the (unique) path from $s$ to $r$ in $T$.

To capture the parity of lengths of paths in a robust manner, we use graphs with edges labeled with elements of $\mathbb{F}_2$. 
Let $G$ be a graph where every edge $e \in E(G)$ is assigned an element $\lambda(e) \in \mathbb{F}_2$. 
With a walk $W$ in $G$ we can associate then the sum of the elements assigned to the edges on $W$ (with multiplicities, i.e., if an edge $e$ appears $c$ times in $W$, then we add $c\cdot\lambda(e)$ to the sum). 
An important observation is that if in $G$ in every closed walk the edge labels sum up to $0$, then for every $u,v \in V(G)$, in every walk from $u$ to $v$ the edge labels sum up to the same
value, depending only on $u$ and $v$. Furthermore, one can in linear time (a) check if every closed walk in $G$ sums up to $0$ and, if this is the case, (b) compute for every $u$ a value
$x_u \in \mathbb{F}_2$, called henceforth the \emph{potential}, such that for every $u,v \in V(G)$ and every walk $W$ from $u$ to $v$ in $G$, the sum of the labels of $W$ equals $x_v - x_u$. 
Indeed, it suffices to take any rooted spanning forest $F$ of $G$, define $x_u$ to be the sum of the labels on the path from $u$ to the root of the corresponding tree in $F$, and check for every
$uv \in E(G)$ if $\lambda(uv) = x_v - x_u$.

\paragraph{Treewidth.} We recall the definitions of treewidth and nice tree decompositions.

\begin{definition}[Tree decomposition and treewidth]
    A tree decomposition of a graph $G$ is a pair $\mathcal{T}=(T,\X)$ where $T$ is a tree whose every node $t$ is associated to a set $X_t \in \X$, called bag of $T$, such that the following conditions hold:
    \begin{itemize}
        \item every vertex of $G$ is contained in at least one bag of $T$;
        \item for every $uv \in E(G)$, there exists $t\in V(T)$ whose bag $X_t$ contains both $u$ and $v$; and
        \item for every $u \in V(G)$, the set of nodes whose bags contain $u$ is connected in $T$.
    \end{itemize}
    The width of a tree decomposition is defined as $\max_{t \in V(T)} |X_t|-1$.
    The treewidth of a graph~$G$, denoted $\tw(G)$, is the minimum possible width of a tree decomposition
    of $G$.
\end{definition}

\begin{definition}[Nice tree decomposition]
    A nice tree decomposition of a graph $G$ is a rooted tree decomposition $(T,\{X_t\}_{t \in V(T)})$
    such that:
    \begin{itemize}
        \item the root and leaves of $T$ have empty bags; and
        \item other nodes are of one of the following types:
        \begin{itemize}
            \item \textbf{Introduce vertex node}: a node $t$ with exactly one child $t'$ such
            that $X_t=X_{t'} \cup \{v\}$ with $v \notin X_{t'}$.
            We say that $v$ is introduced at $t$;
            \item \textbf{Forget vertex node}: a node $t$ with only one child $t'$ such that
            $X_{t}=X_{t'} \setminus \{v\}$ with $v \in X_{t'}$. We say that $v$ is forgotten at $t$;
            and
            \item \textbf{Join node}: a node $t$ with two children $t_1$, $t_2$ such that
            $X_t=X_{t_1}=X_{t_2}$.
        \end{itemize}
    \end{itemize}
    For each node $t$ of the decomposition, we define a \emph{partial} graph $G_t =
    G\left[\bigcup_{s \leq t} X_s\right] - E(G[X_t])$.
\end{definition}
Note that edges of partial graphs appear at forget vertex nodes and that they correspond to adding
edges between the forgotten vertex and its neighbours.

From a tree decomposition $\mathcal{T}=(T,\{X_t\}_{t \in V(T)})$ of $G$ of width $k$, a nice tree
decomposition of width $k$ with $\Oh(k|V(G)|)$ nodes can be computed in time
$\Oh(k^2\cdot\max(|V(T)|,|V(G)|))$, see \cite{kloks1994treewidth}.

\paragraph{Clique-width.} 
A \emph{$k$-labeled} graph is a graph $G$ together with a labeling function $\Gamma: V(G)\rightarrow
[k]$.
For $k$-labeled graphs $H,G$, and integers $i,j\in [k]$, we consider the following operations:
\begin{itemize}
    \item \textbf{Vertex creation:} ${i(v)}$ is the $k$-labeled graph consisting of a single vertex
    $v$ with label $i$;

    \item \textbf{Disjoint union:} ${H \oplus\,G}$ is the $k$-labeled graph consisting of the
    disjoint union of $H$ and $G$;

    \item \textbf{Join:} $\eta_{i\times j}(G)$ is the $k$-labeled graph obtained by adding an edge
    between any pair of vertices one being of label $i$, the other of label $j$, if the edge does
    not exist; and

    \item \textbf{Renaming label:} $\rho_{i\rightarrow j}(G)$ is the $k$-labeled graph obtained by
    changing the label of every vertex labeled $i$ to label $j$: $\forall v \in
    \Gamma^{-1}(\{i\}),\ \Gamma(v) := j$
\end{itemize}
The \emph{clique-width} of a graph $G$, denoted $\cw(G)$, is the least integer $k$ such that a
$k$-labeled graph isomorphic to $G$ can be constructed using these operations.
We call \emph{$k$-expression} of a $k$-labeled graph $G$ a sequence of operations that leads to the
construction of $G$.
Note that such a sequence defines a tree, called \emph{tree associated to the $k$-expression} in the following.

The \emph{linear clique-width} of a graph $G$, denoted $\lcw(G)$, is the least
integer $k$ such that a $k$-labeled graph isomorphic to $G$ can be constructed by a \emph{linear
$k$-expression}, which is a $k$-expression where disjoint union nodes always have one child
that is a vertex creation node.

Note that a it may happen that an edge having its endpoints in two different labels $i$ and $j$ is already present in the graph $G$ before performing the join $\eta_{i\times j}(G)$.
Despite that the edge is not produced twice, the existence of such a situation may be problematic in our algorithms when we only consider a compact representation of $G$.
To circumvent this problem, we can assume that every edge of a graph appears at most once in the join of our given $k$-expressions.
More precisely, when performing a join operation $\eta_{i\times j}(G)$, we can assume that none of the edges in $G$ has its endpoints in $i$ and $j$, respectively.
An expression with such property can be computed in time $\mathcal{O}(k^2 \cdot n)$ from a given arbitrary expression of size $n$.
A proof of this observation is included in appendix. 

Consider a $k$-expression of a $k$-labeled graph $G$, and its associated tree $T$.
For a node $t\in V(T)$, we denote by $T_t$ the subtree of $T$ rooted at $t$, and associate it with
the labeled graph $G_t$ it describes.
For an integer $i\in [k]$, we denote by $V_i(G)$ the set of vertices of label $i$ in $G$.

By an abuse of notations in the following, by ``label $i$'' for a labeled graph $G$ we may refer to
both the integer $i$, or the set $V_i(G)$.

We define \emph{partially $k$-labeled graphs} as labeled graphs with a labeling function $\Gamma(G)
: V(G) \rightarrow [k] \cup \{ \bot \}$ and call \emph{unlabeled} the vertices of
$\Gamma^{-1}(\{\bot\})$.

The problems (together with their parameterization) we consider in this paper are the following.

\begin{center}
    \defproblem{\textsc{Subset Odd Cycle Transversal} (SOCT)}%
    {A graph $G$, a subset $S\subseteq V(G)$, and an integer $p$.}%
    {$\tw(G)$}%
    {Is there a set of at most $p$ vertices hitting every odd cycle of $G$ that contains a vertex of $S$?}

    \vspace{-0.5cm}

    \defproblem{\textsc{Subset Even Cycle Transversal} (SECT)}%
    {A graph $G$, a subset $S\subseteq V(G)$, and an integer $p$.}%
    {$\tw(G)$}%
    {Is there a set of at most $p$ vertices hitting every even cycle of $G$ that contains a vertex of $S$?}

    \vspace{-0.5cm}

    \defproblem{\textsc{Odd Cycle Transversal (OCT)}}%
    {A graph $G$ and an integer $p$.}%
    {$\cw(G)$}%
    {Is there a set of at most $p$ vertices hitting every odd cycle of $G$?}

    \vspace{-0.5cm}

	\defproblem{\textsc{Even Cycle Transversal (ECT)}}%
    {A graph $G$ and an integer $p$.}%
    {$\tw(G)$}%
    {Is there a set of at most $p$ vertices hitting every even cycle of $G$?}

    \vspace{-0.5cm}

    \defproblem{\textsc{Subset Feedback Vertex Set} (SFVS)}%
    {A graph $G$, a subset $S\subseteq V(G)$, and an integer $p$.}%
    {$\tw(G)$ or $\cw(G)$}%
    {Is there a set of at most $p$ vertices hitting every cycle of $G$ that contains a vertex of $S$?}
\end{center}

Given a graph $G$ and a set of vertices $S \subseteq V(G)$, we call $S$-vertex a vertex that is part
of $S$ and we call $S$-path (resp. $S$-cycle) a path that contains at least one $S$-vertex.

We call \emph{nontrivial} a 2-connected multigraph that contains a cycle. Other 2-connected graphs are
the degenerate cases of a single vertex and a bridge, i.e., two vertices connected by a single edge. A
\emph{nontrivial} 2-connected component of a multigraph is a 2-connected component which is a
nontrivial 2-connected multigraph, it is not an isolated vertex or a bridge.

Since our algorithms solve weighted variants of the problems, we will denote by $c:V(G) \rightarrow
\mathbb{Z} \cup \{+\infty\}$ the weight function of the instance. We extend this notation to sets of
vertices with $c(U)= \sum_{v \in U} c(v)$. The unweighted variant corresponds to having $c(v)=1$
for all $v \in V(G)$.

In the context of a dynamic programming algorithm, a state is a tuple of parameters used to index
the table in which computations are done. We denote our table by $d$. We call transition from a set of states
$\mathcal{A}$ to a single state $B$ the action of updating the entry indexed by $B$ based on the values of
states in $\mathcal{A}$. Since we consider only minimizing problems, for a function $f$, such a transition will consist in
applying the operation $$d[B] := \min\{d[B],f(\mathcal{A})\}.$$ We denote this operation by $d[B]
\leftarrow f(\mathcal{A})$ and say that value $f(\mathcal{A})$ is propagated to state $B$.

In the following, we will describe some dynamic programming states with a partially labeled forest
$F$, to be more precise, we describe the state with an arbitrary rooting of such a forest. This
allows for $\Oh(|V(F)|)$ equality testing and preserves all information about the forest. We call
\emph{forest description} such a rooted representation of a forest. A rooted forest can for instance
be encoded by a tuple $(N,f,\ell)$ where $N$ is the number of vertices, $f : [N] \rightarrow [N]$ is
a function where for vertex $u$, $f(u)$ indicates its parent in the rooted forest or $u$ if it is a
root, and $\ell : [N] \rightarrow \mathcal{S}$ is a labeling over set of symbols $\mathcal{S}$ that
can additionally contain a symbol for active vertices. In order to keep notations simple, a forest
description will be denoted like the partially labeled forest it describes.

\section{Hitting even cycles in graphs of bounded treewidth}\label{sec:ect}
\def\ECT{\textsc{Even Cycle Transversal}}
\def\SECT{\textsc{Subset Even Cycle Transversal}}
\def\SOCT{\textsc{Subset Odd Cycle Transversal}}
\def\SFVS{\textsc{Subset Feedback Vertex Set}}
\def\sq{{\small $\square$}}
\def\tw{\textbf{tw}}

In \cite{bergougnoux2020close}, Bergougnoux, Bonnet, Brettell, and Kwon gave a $2^{\Omega(\tw \log
\tw)}n^{\Oh(1)}$ lower bound for \SFVS{}, \SOCT{}, and \ECT{} under the ETH.
We present an algorithm of complexity $2^{\Oh(\tw \log \tw)}n$ for them and \SECT, closing the gap for
\ECT{} and \SECT, and improving on the previous $2^{\Oh(\tw \log \tw)}n^3$ algorithm for \SFVS{} and
\SOCT{} of \cite{bergougnoux2020close}.

Rather than simply giving an algorithm for just SOCT and SECT, we also show how our method gives
less involved algorithms for SFVS and ECT.
All these problems can be seen as looking for a minimum deletion set such that the resulting graph has no
$S$-cycle, no even cycle, no odd $S$-cycle, or no even $S$-cycle. In order to have a common notation,
we will call \sq-cycles the cycles that have to be hit in the problem and \sq-cycle-free the
graphs that do not contain \sq-cycles.

\subsection{Forest representation of \sq-cycle-free graphs}

To transform a \sq-cycle-free graph into a forest, we will replace its nontrivial 2-connected components
with tree structures. We use labeled vertices to store efficiently the properties of these
nontrivial 2-connected components.

We begin by giving a characterisation of nontrivial 2-connected \sq-cycle-free graphs for each problem.
This implies characterisations of \sq-cycle-free graphs.

\begin{lemma}\label{lem:2-cc-charac}
	Let $G$ be a nontrivial 2-connected multigraph.
\begin{enumerate}
	\item $G$ contains no $S$-cycle if and only if it contains no $S$-vertex.

	\item $G$ contains no even cycle if and only if it is an odd cycle.

	\item $G$ contains no odd $S$-cycle if and only if it has one of the following forms:
	\begin{itemize}
		\item $G$ contains no $S$-vertex and is not bipartite
		\item $G$ contains no $S$-vertex and is bipartite
		\item $G$ contains at least one $S$-vertex and is bipartite.
	\end{itemize}

	\item $G$ contains no even $S$-cycle if and only if it has of one of the following forms:
	\begin{itemize}
		\item $G$ contains no $S$-vertex and is not bipartite
		\item $G$ contains no $S$-vertex and is bipartite
		\item $G$ contains at least one $S$-vertex, the connected components of $G-S$ are bipartite, together with
		$S$-vertices they form a cycle: each $S$-vertex has degree 2 and each connected component of
		$G-S$ has outdegree 2. One $S$-cycle is odd.
		We later call \emph{bipartite subcomponents} the connected components of $G-S$.
		This is illustrated in Figure \ref{fig_forest_a}.
	\end{itemize}
\end{enumerate}
\end{lemma}

The first point is immediate, the second was observed in \cite{misra2012parameterized} and the third
in \cite{bergougnoux2020close}. The last point was not known to us and we provide a proof.

\begin{proof}
Suppose that $G$ is a nontrivial 2-connected multigraph containing no even $S$-cycle.

If $G$ contains no $S$-vertex, it is either bipartite or not, leading to the first possible forms.

If $G$ contains an $S$-vertex, it must contain an $S$-cycle $C$ due to being a nontrivial
2-connected multigraph and $C$ is odd because $G$ contains no even $S$-cycle.

\begin{claim}\label{claim:theta}
	If two vertices are connected by three disjoints paths at least two of which are $S$-paths
then two of the paths form an even $S$-cycle.
\end{claim}

\begin{proof}
The three cycles formed by combining the paths are $S$-cycles and they cannot all be odd: if we
denote them $C_1,C_2,C_3$, $|C_3| = |C_1| + |C_2| - 2|C_1 \cap C_2|$.
\end{proof}

Consider a connected component $A$ of $G-V(C)$.

Consider an $S$-vertex $v$ of $A$, because $G$ is a nontrivial 2-connected multigraph, there exist two disjoint paths that
connect $v$ to distinct vertices $a,b$ of $C$, $a$ and $b$ satisfy the conditions of claim
\ref{claim:theta} leading to a contradiction. Hence, $A$ cannot contain an $S$-vertex.

Since $G$ is a nontrivial 2-connected multigraph, there are at least 2 edges between $A$ and distinct
vertices of $C$.
Consider 2 arbitrary distinct such edges, they cut $C$ into two paths $P_1$ and
$P_2$ with extremities $u$ and $v$. Since $A$ is connected, there is a third $u$--$v$ path $P_3$ through
$A$. Only one of $P_1$ and $P_2$ may contain an $S$-vertex, by claim \ref{claim:theta} applied to
$P_1,P_2,P_3$. 
In particular, $u$ and $v$ cannot be $S$-vertices, this implies that the only edges incident to
$S$-vertices in $G$ are edges of cycle $C$, so $S$-vertices have degree 2.

Let $\widetilde{A}$ be the connected component of $G-S$ containing $A$. $\widetilde{A}$ contains a
maximal $S$-free path of $C$ because $A$ is connected to $C$ and cannot be adjacent to
$S$-vertices. $\widetilde{A}$ contains only one maximal $S$-free path of $C$ because otherwise
either we get two edges from $A$ to $C$ that separate $C$ in two $S$-paths and this was excluded in
the previous paragraph, or we have a chord $ab$ in $C$ that connects two distinct maximal
$S$-free-paths and this is excluded by Claim \ref{claim:theta}. In particular, note that this shows
that $\widetilde{A}$ has outdegree $2$ in $G$.

Consider a cycle $C'$ of $\widetilde{A}$, then there are 2 disjoint paths from it to
the $S$-vertices adjacent to $\widetilde{A}$. If they are distinct we can connect them with a disjoint
path via $C$. This constructions contains two $S$-cycles $C$ and $C\Delta C'$ which must both be odd so
$C'$ can only be even. Hence $\widetilde{A}$ contains no odd cycle so it is bipartite.

We can conclude that all connected components of $G-S$ are bipartite and that together with
$S$-vertices they form a cycle.

Conversely, if $G$ contains no $S$-vertex it does not contain any even $S$-cycle. If it is a cycle of
bipartite components and $S$-vertices with one $S$-cycle $C$ being odd, then each $S$-cycle $C'$
goes through all bipartite components and $S$-vertices. Replacing the path of $C$ by the path of
$C'$ in each bipartite component preserves parity because endpoints are unique. We conclude that all
$S$-cycles are odd in $G$.
\end{proof}

\begin{definition}\label{def:underlying-forest}
Given a \sq-cycle-free graph $G$, we define its \emph{underlying forest} $F(G)$ as the graph
obtained from $G$ by modifying independently each nontrivial 2-connected component $C$ as follows:
\begin{enumerate}
	\item For SFVS, remove edges inside $C$ and add an unlabeled vertex adjacent to all
	vertices of $C$.
	\item For ECT, remove edges inside $C$ and add a vertex adjacent to all vertices of
	$C$ and label it “odd cycle”.
	\item For SOCT, remove edges inside $C$ and add a vertex adjacent to all vertices of $C$, label
	it “bipartite” or “not bipartite” based on the property of $C$ and make it an $S$-vertex if $C$
	contains an $S$-vertex.
	\item For SECT, in the two first forms we remove edges inside $C$ and add a vertex
	adjacent to all vertices of $C$ and label it “bipartite” or “not bipartite” based on
	the property of $C$. For the last form, for each bipartite subcomponent in the cycle, we
	remove its edges, add a vertex labeled “internal bipartite” adjacent to its vertices. Then
	remove edges of $C$ incident to $S$, add an $S$-vertex labeled “odd cycle” adjacent to
	$S$-vertices and vertices labeled “internal bipartite”. 
	This is illustrated in Figure \ref{fig_forest_b}.
\end{enumerate}
\end{definition}

\begin{figure}[!h]
 \begin{subfigure}{0.5\linewidth}
  \centering{\includegraphics{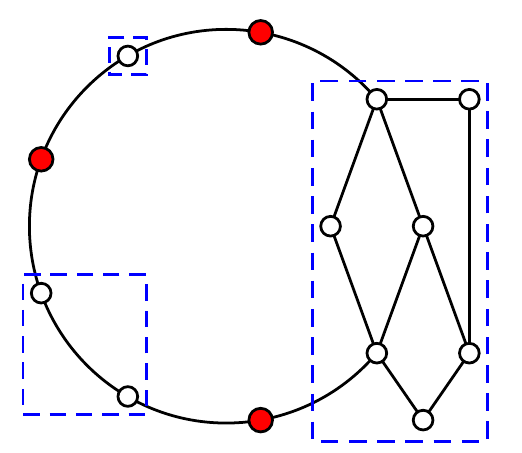}}
  \caption{An example of graph with no even $S$-cycle. \\ The vertices of $S$ are depicted in red. The blue \\ boxes denote the bipartite subcomponents.}
  \label{fig_forest_a}
 \end{subfigure}
 \begin{subfigure}{0.5\linewidth}
  \centering{\includegraphics{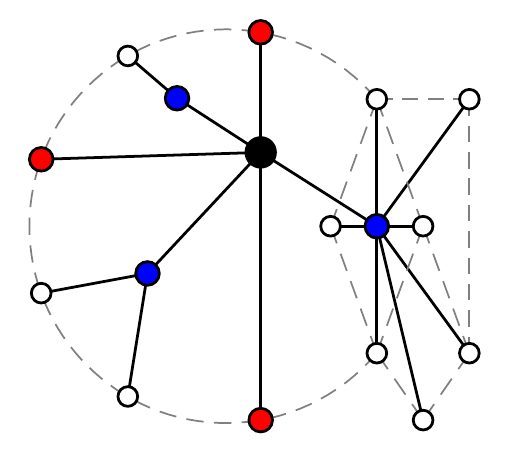}}
  \caption{The underlying forest we build from the graph on Figure \ref{fig_forest_a}. The “internal
  bipartite” vertices are depicted in blue and the “odd cycle” vertex is black.}
  \label{fig_forest_b}
 \end{subfigure}
 \caption{The last form of SECT: ``internal bipartite'' vertices}
 \label{fig_forest}
\end{figure}

Observe that, because labeled vertices are only introduced by this underlying forest, to each
labeled vertex $v$, we can associate a nontrivial 2-connected component $C$: the one that resulted
in the creation of $v$. Observe also that for a path $P$ between two unlabeled vertices, if it contains
a labeled vertex, then it contains a vertex of its associated component before it on $P$ and another
vertex of its associated component after it on $P$.

We now introduce reduction rules that allow us to maintain a simplified description of underlying
forests relatively to a set of \emph{active} vertices. Vertices that are not active are called
\emph{inactive}. Theses rules and this terminology are derived from \cite{bergougnoux2020close}.

\begin{definition}\label{def:red-forest}
Given a \sq-cycle-free graph $G$, its underlying forest $F(G)$ and subset of active vertices
$X$, a reduced underlying forest $F_r$ is obtained by applying exhaustively the following rules
on $F(G)$:
\begin{itemize}
	\item Delete inactive vertices of degree at most one.
	\item For each maximal path $P$ with internal inactive vertices of degree 2, we replace it with
	a path $P'$ with same endpoints, such that $P'$ contains exactly one occurrence of each label
	present in $P$ and a single $S$-vertex if $P$ contained one, where endpoints are considered to
	be contained in $P$ and $P'$.
\end{itemize}
For SECT we add another rule: if a maximal path with internal inactive vertices of degree 2
contains at least 2 vertices labeled “internal bipartite” but no vertex labeled “odd cycle”, we
keep 2 occurences of the label “internal bipartite”.

The set of reduced underlying forests obtained from $F(G)$ with active vertices $X$ is denoted $F(G,X)$.
\end{definition}

Observe that a reduced forest is not unique, however properties that we will show on them will not
depend on the choice of representative. Furthermore, in our dynamic programming states, we use
rooted representations that are not necessarily unique either for each reduced underlying forest.

In the next lemma we show that a reduced forest has bounded size. 
On a maximal path with internal inactive vertices of degree 2 and 2 vertices labeled ``internal
bipartite'', there is no vertex labeled ``odd cycle''. Hence the maximum number of vertices on such
a path after the reduction is also bounded by the number of labels.

\begin{lemma}\label{lem:red-forest-size}
	In a problem using $K$ label symbols (including $S$-membership), $F \in F(G,X)$ has at most
	$(K+1)(2|X|-2)+1$ vertices.
\end{lemma}

\begin{proof}
	By the first reduction rule, all leaves of $F$ are active vertices. We use the following
	result that can be proved by induction.
	\begin{claim}
		A non-empty tree with $p$ leaves and internal degree at least 3 has at most $2p-1$ vertices.
	\end{claim}
	Consider the forest $F'$ obtained from $F$ by replacing maximal paths with inactive
	internal vertices of degree 2 by edges. Let $p$ denote the number of leaves of $F'$ and $q$
	denote the number of active vertices of degree 2. We deduce from the claim that $F'$ has at most
	$2p+q-1$ vertices. $p+q \leq |X|$ because $p$ and $q$ are cardinals of disjoint parts of $X$, so
	$F'$ has at most $2|X|-1$ vertices. Then $F'$ has at most $2|X|-2$ edges which correspond to
	paths in $F$ with at most $K$ internal vertices. Summing up, $F$ has at most
	$K(2|X|-2) + 2|X| -1 = (K+1)(2|X|-2)+1$ vertices.
\end{proof}

The crucial property preserved by a reduced forest is the following.
\begin{claim}\label{claim:ect-red-prop}
For $F \in F(G,X)$, for each pair of active vertices $u$ and $v$, there is a path between them in
$F(G)$ if and only if there is a path between them in $F$. For each label symbol, the path in
$F(G)$ contains a vertex with this symbol if and only if the path in $F$ contains a vertex with
this symbol.
\end{claim}

We immediately deduce the following lemma.
\begin{lemma}\label{lem:ect-red-prop}
For $F \in F(G,X)$, for each pair of active vertices $u$ and $v$, there is a path between them in
$G$ if and only if there is a path between them in $F$. For each type of nontrivial 2-connected
component, a $u$--$v$ path in $G$ goes through at least one such component if and only if the $u$--$v$ path
in $F$ contains a vertex with the corresponding label symbol (“internal bipartite” counts for the
bipartite subcomponent but also the $S$-cycle containing it). 
There exists a $u$--$v$ path in $G$ containing an $S$-vertex if and only if there exists a $u$--$v$ path
in $F$ containing an $S$-vertex or a vertex labeled “internal bipartite”.
If there is a $u$--$v$ path in $F$, every unlabeled vertex that is on the $u$--$v$ path in $F$ is also on
all $u$--$v$ paths in $G$.
\end{lemma}

A property that is not preserved by a reduced forest is the length of paths. Since we are only
interested in parity, we maintain a $\mathbb{F}_2$-labeling $\alpha$ of edges.
We say that $\alpha$ is a \emph{valid} $\mathbb{F}_2$-labeling of $F \in F(G,X)$ if, there exists
$\beta$ a $\mathbb{F}_2$-labeling of the edges of $F(G)$ such that edges incident to vertices
labeled ``bipartite'' or ``internal bipartite'' are labeled $0$ for one side of the bipartition
and $1$ for the other side, edges incident to other labeled vertices are labeled $0$, and edges
between unlabeled vertices are labeled $1$, and for each edge $uv$ of $F$, its label is the sum of
labels on the edges of the $u$--$v$ path in $F(G)$.
During the application of reduction rules, each edge is given as its label the sum of labels of the
path that was connecting its endpoints.

\begin{lemma}\label{lem:red-color}
For $F \in F(G,X)$, for each pair of active vertices $u$ and $v$ connected in $G$, all $u$--$v$ paths in
$G$ have same parity if and only if the path between $u$ and $v$ in $F$ contains no vertex with
label symbol “odd cycle”, “not bipartite” or “internal bipartite”.
Furthermore, when this condition is satisfied, the parity of the paths in $G$ is given by the sum of
labels on the edges of $F$.
\end{lemma}

\subsection{Reduced forest joins}

The main technical engine of our algorithms is the following join operation.
\begin{lemma}\label{lem:merge-F}
There exists a polynomial-time algorithm that, for every
pair of \sq-cycle-free graphs $G_1$ and $G_2$ with $V(G_1) \cap V(G_2) = X$, 
given on input two reduced forests with valid $\mathbb{F}_2$-labelings $(F_1,\alpha_1)$
and $(F_2,\alpha_2)$, with $F_1 \in F(G_1,X)$ and $F_2 \in F(G_2,X)$, decides whether $G_1 \cup G_2$
is \sq-cycle-free and, in case of a positive answer, computes a reduced forest $F \in F(G_1 \cup G_2, X)$
and, except for the \textsc{SFVS} problem, a valid $\mathbb{F}_2$-labeling $\alpha$.
\end{lemma}

This section is devoted to the proof of Lemma~\ref{lem:merge-F}. 
Denote $H = F_1 \cup F_2$. Note that the common vertices of $F_1$ and $F_2$ are exactly the vertices
of $X$ and they are unlabeled vertices.

We start with some claims about the correspondence between $H$ and $G_1 \cup G_2$.
We partition 2-connected components of $H$ into \emph{blobs}: a blob is a maximal union of
nontrivial 2-connected components of $H$ such that every two 2-connected components of a blob can be
connected in $H$ with a path whose all cutvertices are labeled vertices of $H$.
The reason behind this is that although a labeled vertex may be a cutvertex in $H$, since it is
labeled, it corresponds to a nontrivial 2-connected component of $G_1$ or $G_2$ by definition of
the underlying forest.
Observe that if a labeled vertex is in a cycle of $H$ then in the
union of underlying forests (not reduced) the cycle obtained by expanding contracted paths contains
two vertices of the labeled vertex' corresponding component, meaning that vertices along this cycle
are in the same 2-connected component in $G_1 \cup G_2$. 

Given a path $P$ between active vertices in $F_i$, we define a \emph{lift} of this path in $G_i$ as follows.
Because there is a path between the endpoints of $P$ in $F_i$, by Lemma \ref{lem:ect-red-prop}, there is
a path $P'$ between the two vertices in $G_i$. Note that the nontrivial 2-connected components and
cut vertices on the path are already determined. Inside the nontrivial 2-connected components, the
path is completed arbitrarily except for the case of odd cycles in SECT where a path without
$S$-vertices is chosen if it exists. The labeled vertices of $P$ are not useful to define $P'$,
however by construction of $F_i$, they give information about the lift.

\begin{claim}\label{claim:lift-walk}
	Let $W$ be a closed walk in a graph $G$, such that for every vertex $v$ on $W$, at most one visit of $W$
	in $v$ has the property that the in- and out-edge belong to distinct 2-connected components.
	Then, all edges of $W$ lie in one 2-connected component of $G$.
\end{claim}

\begin{proof}
	Assume towards a contradiction that $W$ contains edges of at least two 2-connected components, then it contains a cut
	vertex $v$ of $G$. $v$ is a vertex of $W$, so it is visited once by $W$ going from 2-connected
	component $C_1$ to $C_2$. Since it is a cut	vertex of $G$ and $W$ is a closed walk, there must
	be a second visit going back to $C_1$. We conclude that $v$ violates the property that at most
	one visit of $W$ in $v$ has in- and out-edges belonging to distinct 2-connected components.	
\end{proof}

\begin{claim}\label{claim:H-to-G}
	If $C$ is a blob of $H$, then there exists $C'$ a nontrivial 2-connected component of $G_1 \cup
	G_2$ that contains all unlabeled vertices of $C$ and vertices in components represented by
	labeled vertices of $C$.
\end{claim}

\begin{proof}
	First we consider a cycle $C$ in $H$, it must consist of a succession of paths in $F_1$ and
	paths in $F_2$ because $F_1$ and $F_2$ are forests. From each such maximal path of $F_i$
	we can deduce a lift, now the succession of lifts defines a closed walk $W$. 
	Let $v$ be a repeated vertex in $W$.
	Because $C$ is a simple cycle in $H$, $v$ is visited at most once in $C$ so other visits can only be
	caused by $C$ visiting a labeled vertex $u$ associated with the nontrivial 2-connected
	component containing $v$. In such other visit, the vertices from which $W$ enter the
	2-connected component represented by $u$ are distinct from $v$ because it cannot be visited more
	than once in $C$. This means that the in- and out-edges of such visit are in the nontrivial
	2-connected component represented by $u$. We conclude that conditions of Claim
	\ref{claim:lift-walk} are satisfied and deduce that $W$ is contained in one 2-connected
	component of $G_1 \cup G_2$.

	Now if two cycles $C_1$ and $C_2$ of $H$ are not edge disjoint, then there exist two distinct
	vertices $a,b$ contained by both walks $W_1$ and $W_2$ obtained by lifts from $C_1$ and $C_2$.
	This means that the 2-connected components of $G_1 \cup G_2$ containing $W_1$ and $W_2$ have two
	distinct vertices in common, so it must be the same 2-connected component.

	If a labeled vertex $u$ is in a cycle $C$ of $H$, then all of the vertices that are contained
	in the nontrivial 2-connected component $B$ represented by $u$ are inside the same maximal
	2-connected component of $G_1 \cup G_2$ as $W$ the walk obtained by lifts from $C$. Indeed, if
	$u$ is in $C$, then $W$ must contain two distinct vertices of $B$, and since $B$ is 2-connected
	it must be contained in the same 2-connected component of $G_1 \cup G_2$ as $W$.

	Finally, if two cycles $C_1$ and $C_2$ of $H$ are edge disjoint but share a labeled vertex $u$,
	then consider the walks $W_1$ and $W_2$ in $G_1 \cup G_2$ obtained by lifts from $C_1$ and
	$C_2$. Each walk must contain two distinct vertices of the 2-connected subgraph $B$ represented
	by $u$. Hence, the 2-connected component of $G_1 \cup G_2$ containing $B$ also contains $W_1$
	and $W_2$.
	
	We conclude from the previous points that for $B$ a blob of $H$, all of its vertices and the
	vertices contained in its labeled vertices' associated nontrivial 2-connected component are in
	the same maximal 2-connected component of $G_1 \cup G_2$.
\end{proof}

For a cycle $C$ in $G_1 \cup G_2$ that contains edges of both $G_1$ and $G_2$, we construct a walk $W$ in $H$ as follows.
Since $C$ contains edges of both $G_1$ and $G_2$, it can be decomposed into paths $P_1,P_2,\ldots,P_\ell$, with $\ell$ even,
and paths $P_1,P_3,P_5,\ldots$ lying in $G_1$ and $P_2,P_4,P_6$ lying in $G_2$; the path $P_i$ has both endpoints being active vertices,
one being also the endpoint of $P_{i-1}$ and one being also the endpoint of $P_{i+1}$.
By Lemma~\ref{lem:ect-red-prop}, for each $P_i$ there is a corresponding path $P_i'$ between the same endpoints in $F_1$ or $F_2$.
The concatenation of the paths $P_i'$ is a closed walk in $H$ which we call \emph{the contracted walk of $C$}. 

\begin{claim}\label{claim:G-to-H}
	If $C$ is a cycle of $G_1\cup G_2$ that contains at least one edge of $G_1$ and at least one edge of $G_2$,
  then its contracted walk in $H$ visits more than once only labeled vertices. Furthermore, it is contained in a single blob of $H$.
\end{claim}
\begin{proof}
	Because vertices of $C$ appear at most once and vertices of the contracted walk $W$ that are not
	labeled are vertices of $C$ appearing as many times, we conclude that only labeled vertices
	can be visited more than once. Now only labeled vertices can be cut vertices in $W$ and by
	definition of blobs they do not separate distinct blobs.
\end{proof}

We now move to the description of the algorithm of Lemma~\ref{lem:merge-F}. 
For all problems, start by computing the 2-connected components of $H$ and blobs.
Then, proceed as follows, depending on the problem at hand. 

\begin{description}
\item[SFVS.] Check that no blob of $H$ contains an $S$-vertex.
Then, return any element of $F(H,X)$ as the desired forest. (There is no need to compute $\alpha$ for \textsc{SFVS}.)
\item[ECT.]
For each blob $B$ of $H$, proceed as follows.
Reject if $B$ is not a cycle or contains a vertex with label “odd cycle”.
Also reject if the sum of labels of edges of the cycle is $0$. 
Otherwise, the edges of $B$ are removed and a vertex adjacent to all vertices of $B$ is added,
with a label “odd cycle”. The new edges are given label $0$.
\item[SOCT.]
For each blob $B$ of $H$, proceed as follows.
Check if $B$ contains a cycle that does not sum to $0$, if not, compute the
bipartition of vertices based on their potential.
Reject if $B$ contains an $S$-vertex and either a cycle that does not sum to $0$ or a vertex labeled “not bipartite”.
Otherwise, the internal edges of $B$ are removed, a vertex adjacent to unlabeled vertices of $B$ is
added, labeled vertices of $B$ are identified to it. The new vertex is an $S$-vertex if $B$ contains
an $S$-vertex and is labeled “bipartite” if $B$ contained no cycle that does not sum to $0$ nor
vertex labeled “not bipartite”, otherwise, it is labeled “not bipartite”.
If we add a vertex labeled ``bipartite'', two cases arise for edges incident to it. First, if the
other endpoint is in $B$, it is labeled with the computed potential of this endpoint. Second, if the
other endpoints is not in $B$, it corresponds to an edge that was incident to a former labeled
vertex $v$ in $B$, we add the potential of $v$ to the previous label of the edge.
If we add a vertex labeled ``not bipartite'', we label all edges with $0$.
\item[SECT.]
For each blob $B$ of $H$, proceed as follows.
First, we consider the case when $B$ contains no $S$-vertex. 
Check with depth-first search if $B$ contains a cycle that does not sum to $0$, if not, compute the
bipartition of vertices based on their potential.
Reject if there is more than one vertex labeled ``internal bipartite'' in $B$, 
or there is a vertex labeled ``internal bipartite'' and a cycle that does not sum to $0$,
or there is a vertex labeled ``internal bipartite'' and a vertex labeled ``not bipartite''. 
Otherwise, add a vertex, make it adjacent to unlabeled vertices
of $B$ and identify labeled vertices of $B$ to it. 
This additional vertex is labeled ``bipartite'' if $B$ contained no cycle that does not sum to $0$
nor vertex labeled ``not bipartite'' or ``internal bipartite'', it is labeled ``internal bipartite''
if it contained a vertex labeled ``internal bipartite'', and it is labeled ``not bipartite''
otherwise.
If we add a vertex labeled ``bipartite'' or ``internal bipartite'', two cases arise for edges
incident to it. First, if the other endpoint is in $B$, it is labeled with the computed potential of
this endpoint. Second, if the other endpoints is not in $B$, it corresponds to an edge that was
incident to a former labeled vertex $v$ in $B$, we add the potential of $v$ to the previous label of
the edge. If the additional vertex is labeled differently, we label all edges with $0$.

Consider now the case when $B$ contains an $S$-vertex.
Reject if one of the following cases occur: $B$ contains a vertex labeled ``not bipartite'', ``odd cycle'', or ``internal bipartite'', 
there is an $S$-vertex in $B$ that is of degree more than $2$, or 
there is a connected component of $B-S$ where the number of edges with exactly one endpoint in the
said component is more than $2$.
Otherwise, pick some $S$-vertex $v$ in $B$. We define $\widetilde{B}$ as the graph obtained from $B$
by subdividing $v$ in two vertices separated by a new edge labeled $1$ and each adjacent to one of
the two edges that were adjacent to $v$.
Check with depth-first search if $\widetilde{B}$ contains a cycle that does not sum to $0$, if not,
compute the bipartition of vertices based on their potential.
If we did not reject, for each connected component $B'$ of of $B-S$, 
a vertex adjacent to unlabeled vertices of $B'$ and labeled
``internal bipartite'' is added and labeled vertices of $B'$ are identified with it.
Then edges of $B$ incident to its $S$-vertices are also removed and a vertex labeled ``odd cycle''
adjacent to $S$-vertices of $C$ and to the new vertices labeled ``internal bipartite'' is added.
For an edge $e$ incident to a vertex labeled ``internal bipartite'' but not to a vertex labeled ``odd
cycle'', either $e$ is labeled with the computed potential of its other endpoint in $B$, or the
other endpoint is not in $B$, so $e$ used to be incident to labeled vertex $v$ in $B$, then we add
the potential of $v$ to the edge's previous label. Other edges are labeled $0$.
\end{description}

We conclude by reducing the forest obtained from $H$ and denote by $F$ the resulting forest.

\begin{claim}\label{claim:join-reject}
	The join processing rejects if and only if $G_1 \cup G_2$ is not \sq-cycle-free.
\end{claim}

\begin{proof}
	For SFVS, if the join processing rejects then there is a blob of $H$ containing an $S$-vertex,
	by Claim \ref{claim:H-to-G}, there is a nontrivial 2-connected component containing an
	$S$-vertex in $G_1 \cup G_2$ so there is an $S$-cycle.  

	Conversely, if $G_1 \cup G_2$ is not $S$-cycle-free, since $G_1$ and $G_2$ are $S$-cycle-free,
	it contains an $S$-cycle $C$. We consider the contracted walk of $C$ in $H$. By Claim
	\ref{claim:G-to-H}, this walk is contains an $S$-vertex in $H$ and is contained in a blob of $H$,
	causing the join to reject.

	For ECT, if the join processing rejects one the following happens. 
	\begin{itemize}
	\item There is a blob $B$ that is not a cycle in $H$, then by Claim \ref{claim:H-to-G}, there is a
	nontrivial 2-connected component $C$ of $G_1 \cup G_2$ containing its vertices and because it
	each path in $B$ can be lifted in a walk in $C$, there is more than one cycle in $C$. 
	\item There is a blob $B$ that contains a vertex labeled “odd cycle”, then by Claim
	\ref{claim:H-to-G}, there is a 2-connected component of $G_1 \cup G_2$ containing an odd cycle
	and a vertex that is not in this cycle.
	\item The sum of labels on edges of the cycle is $0$ in a blob of $H$ that is a cycle, then
	there must be a non trivial 2-connected component containing its vertices in $G_1 \cup G_2$ by
	Claim \ref{claim:H-to-G} and it must be bipartite, so there is an even cycle in $G_1 \cup G_2$.
	\end{itemize}
	In all cases, we conclude that $G_1 \cup G_2$ is not even-cycle-free with Lemma
	\ref{lem:2-cc-charac}.

	Conversely, if there exists an even cycle in $G_1 \cup G_2$, since it is not contained in $G_1$
	or $G_2$, it contains edges of both, by Claim \ref{claim:G-to-H}, its contracted walk $W$ visits
	more than once only labeled vertices. First, if $W$ contains a labeled vertex it causes a
	reject. Otherwise, $W$ must be a cycle of $H$. Consider the blob containing $W$, either it is
	not a cycle, causing a reject, or it is a cycle, hence it is exactly $W$ which cannot sum to $0$
	because it represents an even cycle, this also causes a reject. 	

	For SOCT, if the join processing rejects then there must exist a blob $B$ of $H$ that contains an
	$S$-vertex $v$ and either a cycle that does not sum to $0$ or a vertex labeled “not bipartite”,
	by Claim \ref{claim:H-to-G} there is a 2-connected component $C$ of $G_1 \cup G_2$ containing
	$v$ but also an odd cycle: if $B$ contains a vertex labeled “not bipartite”, then $C$ contains
	an odd cycle contained in $G_1$ or $G_2$, otherwise $B$ contains a cycle that does not sum to
	$0$ which means that $C$ is not bipartite by Lemma \ref{lem:red-color}. From Lemma
	\ref{lem:2-cc-charac} we conclude that $G_1 \cup G_2$ is not odd-$S$-cycle-free.

	Conversely, if $G_1 \cup G_2$ contains an odd $S$-cycle, since it is not contained in $G_1$ or
	$G_2$, by Claim \ref{claim:H-to-G}, there must be a corresponding contracted walk $W$ in $H$,
	which is contained by a blob $B$ of $H$. If the blob $B$ contains a vertex labeled “not
	bipartite”, it causes a reject, otherwise the walk $W$ in $H$ encodes parity correctly 
	by Lemma \ref{lem:red-color}, so a cycle that does not sum to $0$ will be found, causing a reject.

	For SECT, if the join processing rejects then there must exist a blob $B$ of $H$
	that does not respect a condition, let $C$ denote the nontrivial 2-connected component of $G_1
	\cup G_2$ containing its vertices (Claim \ref{claim:H-to-G}). 
	We first suppose that $B$ contains no $S$-vertex.
	\begin{itemize}
	\item If $B$ contains two vertices labeled “internal bipartite”, then either $C$ contains only one
	odd-$S$-cycle component but then we added an $S$-free path between two of its bipartite
	subcomponents, or $C$ contains two odd-$S$-cycle
	components. In both cases, $C$ is not of a form that is possible for $G_1 \cup G_2$
	even-$S$-cycle-free.
	\item If $B$ contains a vertex labeled “internal bipartite” and a vertex labeled “not bipartite”
	or a cycle that does not sum to $0$, then $C$ contains an odd-$S$-cycle component and an odd
	cycle (either from the contained component that is not bipartite, or the cycle that does not sum
	to $0$). In both cases, $C$ is not of a form that is possible for $G_1 \cup G_2$
	even-$S$-cycle-free. 
	\end{itemize}
	We now suppose that $B$ contains an $S$-vertex.
	\begin{itemize}
	\item If $B$ contains a vertex labeled “not bipartite” or “odd cycle”, then $C$ cannot be of the form
	that is possible for $G_1 \cup G_2$ even-$S$-cycle-free.
	\item If $B$ contains a vertex labeled internal bipartite then either we have added a path containing
	an $S$-vertex between two vertices of the same bipartite subcomponent of an odd-$S$-cycle
	component contained in $G_1$ or $G_2$, or we have added a path between two existing bipartite
	subcomponents. In both cases $C$ cannot be of the form that is possible for $G_1 \cup G_2$
	even-$S$-cycle-free.
	\item If $B$ contains an $S$-vertex that is of degree more than two, or a connected component of
	$B-S$ has more than 2 edges with exactly one endpoint in the said component, then the same holds
	in $C$. This contradicts the form of a nontrivial 2-connected
	component containing an $S$-vertex for $G_1 \cup G_2$ even-$S$-cycle free
	\item If $\widetilde{B}$ contains a cycle that does not sum to $0$, then there must exist an even
	$S$-cycle in $G$ by construction of $\widetilde{B}$ and Lemma \ref{lem:red-color}. This
	immediately implies that $G_1 \cup G_2$ is not even-$S$-cycle-free.
	\end{itemize}
	
	Conversely, if $G_1 \cup G_2$ contains an even $S$-cycle $C$, then because $G_1$ and $G_2$ don't
	contain it, by Claim \ref{claim:G-to-H}, the corresponding contracted walk $W$ in $H$ is in a
	single blob $B$ of $H$. $B$ contains an $S$-vertex from containing $W$. Either, $B$ contains
	labeled vertices that are not labeled bipartite or does not respect the degree conditions,
	causing a reject, or it contains only labeled vertices with label “bipartite” which means that
	the parity of paths in $B$ is correctly encoded (Lemma \ref{lem:red-color}), and all $S$-cycles
	contain all $S$-vertices from degree properties, in particular, we can deduce that any
	$S$-vertex chosen to be subdivided in $\widetilde{B}$ would be in $C$, and then the cycle has
	even length which will cause a the existence of a cycle that does not sum to $0$ in
	$\widetilde{B}$, causing a reject.
\end{proof}

\begin{claim}\label{claim:join-forest}
	The computed forest is in $F(G_1 \cup G_2,X)$ and the computed $\mathbb{F}_2$-labeling $\alpha$
	of its edges is valid for it, if $G_1 \cup G_2$ is \sq-cycle-free.
\end{claim}

\begin{proof}
	Consider a nontrivial 2-connected component $C$ of $G_1 \cup G_2$, it may be contained in $G_1$
	or $G_2$, in this case its active vertices have already been connected to a labeled vertex
	representing $C$ so there is nothing more to do. If $C$ contains edges from both $G_1$ and $G_2$
	then there must be a blob $B$ in $H$ that contains part of the active vertice of $C$. By
	construction, active vertices of $C$ will be adjacent to the labeled vertex representing their
	component: they can only be in $B$ or adjacent to a labeled vertex in $B$ and labeled vertices
	were identified to the new vertex. 
	Furthermore, one can check that the choice of labels correspond to the forms of the components
	of $G_1 \cup G_2$ based on the information stored in the labels of $F_1$ and $F_2$, and no
	information can be missing (Lemmata \ref{lem:ect-red-prop} and \ref{lem:red-color}).
	
	The vertices that were removed in $F(G_1)$ and in $F(G_2)$ to obtain $F_1$ and $F_2$ must still
	be removed in $F(G_1 \cup G_2)$. In particular when a path of inactive vertices of degree 2
	leads to the creation of a cycle, its inactive vertices become leaves and can then all be
	simplified. Because after producing a forest we reduce it, there cannot be additional reductions
	to be performed, hence the computed forest is in $F(G_1 \cup G_2,X)$.

	It is straightforward to check that $\alpha$ is a valid $\mathbb{F}_2$-labeling of $F$. 
\end{proof}

We conclude the proof of Lemma~\ref{lem:merge-F} with an observation 
that it is straightforward to implement the discussed algorithm to run
in time polynomial in the input size.
Finally, note that the input size is of size $\Oh(|X|)$. 

\subsection{Algorithm}

We now describe our dynamic programming algorithm on a nice tree decomposition $(T,\{X_t\}_{t \in
V(T)})$ of graph $G$. It consists of a bottom-up computation with states $(t,Y,F,\alpha)$ where $t
\in V(T)$, $Y\subseteq X_t$ is the set of undeleted vertices of $X_t$, $F$ is a labeled forest
description with active vertices $Y$, and $\alpha$ is a $\mathbb{F}_2$-labeling of edges of $F$. We
denote by $d[t,Y,F,\alpha]$ the cell of the table corresponding to this state. We call a state
\emph{reachable} if its cell is updated at least once by a transition.

We call a state $(t,Y,F,\alpha)$ \emph{admissible} if there exists $U \subseteq V(G_t)$ such that
$Y=X_t \setminus U$, $G_t-U$ is \sq-cycle-free, $F$ is a forest description of a member of
$F(G_t-U,Y)$, $\alpha$ is a valid 2-labeling of $F$ and $d[t,Y,F,\alpha]=c(U)$.

To prove the correctness of the algorithm we will prove that reachable states are admissible and
that for each $t \in V(T)$ and $U \subseteq V(G_t)$, if $G_t-U$ is \sq-cycle-free there exists a
state with value at most $c(U)$. The optimal transversal weight will be in
$d[r,\varnothing,\varnothing,\varnothing]$ where $r$ is the root of the decomposition.

We now describe the computations for each node $t$ of the nice tree decomposition based on its type.
\begin{enumerate}
	\item \textbf{Leaf node.} We set $d[t,\varnothing,\varnothing,\varnothing]=0$

	\item \textbf{Introduce vertex node.} Let $t'$ denote the child node of $t$ and $v$ be the
	introduced vertex. 
	For each reachable state $(t',Y',F',\alpha')$, we have two transitions representing the choice
	of deleting the vertex or not: $$d[t,Y',F',\alpha'] \leftarrow d[t',Y',F',\alpha'] + c(v)$$
	$$d[t',Y'\cup\{v\},F,\alpha'] \leftarrow d[t',Y',F',\alpha']$$ where $F$ is
	obtained from $F'$ by adding an isolated active vertex $v$.

	\item \textbf{Forget vertex node.} Let $t'$ denote the child node of $t$ and $v$ be the
	forgotten vertex. For each reachable state $(t',Y',F',\alpha')$, if $v \notin Y'$ then the
	transition is simply: $$d[t,Y',F',\alpha'] \leftarrow d[t',Y',F',\alpha']$$ 
	If $v \in Y'$, we perform a join processing of $(F',\alpha')$ and $(H,\beta)$ where $H$ is the
	union of a star graph with internal vertex $v$ and leaves $N_G(v) \cap Y'$ and isolated vertices
	for $Y'\setminus N_G[v]$, and $\beta$ its edges to $1$. 
	If the join processing does not reject and returns $(\widetilde{F},\widetilde{\alpha})$, we
	obtain $(F,\alpha)$ from $(\widetilde{F},\widetilde{\alpha})$ by making $v$ inactive and
	applying reduction rules, then have transition: $$d[t,Y' \setminus \{v\},F,\alpha] \leftarrow
	d[t',Y',F',\alpha']$$

	\item \textbf{Join node.} Let $t_1$ and $t_2$ denote the two children of $t$. For each pair of
	reachable states $(t_1,Y,F_1,\alpha_1)$ and $(t_2,Y,F_2,\alpha_2)$, we perform a join process of
	$(F_1,\alpha_1)$ and $(F_2,\alpha_2)$. If the join isnt rejected, we obtain $(F,\alpha)$ and
	have transition $d[t,Y,F,\alpha] \leftarrow d[t_1,Y,F_1,\alpha_1] + d[t_2,Y,F_2,\alpha_2]$.
\end{enumerate}

\begin{lemma}\label{lem:ect-sound}
	All reachable states are admissible.
\end{lemma}

\begin{proof}
	By induction on $T$.
\begin{enumerate}
	\item If $t$ is a leaf node, then choosing $U=\varnothing$ we have that
	$(t,\varnothing,\varnothing,\varnothing)$ is admissible.
	
	\item If $t$ is an introduce vertex node with child $t'$, introducing vertex $v$, then
	for $(t,Y,F,\alpha)$ a reachable state, there exists the state that gave it optimal value
	$(t',Y',F',\alpha')$. By induction hypothesis applied to $(t',Y',F',\alpha')$, there exists $U'$
	such that $Y'= X_{t'}\setminus U'$, $G_{t'}-U'$ is \sq-cycle-free, $F'$ is a forest description
	of a member $F(G_{t'}-U',Y')$, $\alpha'$ is a valid $\mathbb{F}_2$-labeling of $F'$ and
	$d[t',Y',F',\alpha']=c(U')$.
	If $v \notin Y$, then we can deduce the transition that was used, we set $U= U' \cup \{v\}$,
	since $Y'=X_{t'} \setminus U'$ and $X_t=X_{t'} \cup \{v\}$, we have $Y=X_t \setminus U$. We also
	have $d[t,Y,F,\alpha']=d[t',Y',F',\alpha']+c(v)=c(U')+c(v)=c(U)$. $G_t-U = G_{t'}-U'$, $F'=F$,
	$Y'=Y$, so $G_t-U$ is \sq-cycle-free, $F$ is a forest description of a member of $F(G_t-U,Y)$
	and $\alpha$ is a valid $\mathbb{F}_2$-labeling of $F$.
	If $v \in Y$, then we deduce the transition and set $U=U'$. then $Y=Y' \cup \{v\}$ and since
	$Y'=X_{t'}\setminus U'$ and $X_t=X_{t'} \cup \{v\}$, we have $Y = X_t \setminus U$. We also have
	$d[t,Y,F,\alpha]=d[t',Y',F',\alpha']=c(U')=c(U)$. $G_t-U$ is \sq-cycle-free because $G_{t'}-U'$
	is and $v$ is isolated. $F$ is a forest description of $F(G_t-U,Y)$ by construction, and $\alpha$
	is a valid $\mathbb{F}_2$-labeling of $F$ because $v$ is an isolated vertex.

	\item If $t$ is a forget vertex node with child $t'$, forgetting vertex $v$, then for
	$(t,Y,F,\alpha)$ a reachable state, there exists the state that gave it optimal value
	$(t',Y',F',\alpha')$. By induction hypothesis applied to $(t',Y',F',\alpha')$, there exists $U$
	such that $Y'= X_{t'}\setminus U$, $G_{t'}-U$ is \sq-cycle-free, $F'$ is a forest description
	of a member of $F(G_{t'}-U,Y')$, $\alpha'$ is a valid $\mathbb{F}_2$-labeling of $F'$ and
	$d[t',Y',F',\alpha']=c(U)$. In all cases, $d[t,Y,F,\alpha]=c(U)$.
	If $v \notin Y'$, then we can deduce the transition that was used, and there is nothing to show
	since $Y=Y'$, $G_t-U=G_{t'}-U$, $F=F'$ and $\alpha=\alpha'$.
	If $v \in Y'$, then we can deduce the transition that was used. We have $Y=Y'\setminus \{v \} =
	X_t \setminus U$. $G_t-U = G_{t'}-U \cup H$ so, by Claim \ref{claim:join-reject}, $G_t-U$ is
	\sq-cycle-free. By Claim \ref{claim:join-forest}, $F$ is a forest description of a member of
	$F(G_t-U,Y)$ and $\alpha$ is a valid $\mathbb{F}_2$-labeling of $F$.

	\item If $t$ is a join node with children $t_1$ and $t_2$, then for $(t,Y,F,\alpha)$ a
	reachable state, there exist states that gave it the optimal value $(t_1,Y,F_1,\alpha)$ and
	$(t_2,Y,F_2,\alpha_2)$. By induction hypothesis applied to these states, for $i \in \{1,2\}$,
	there exists $U_i$ such that $Y=X_{t_i} \setminus U_i$, $G_{t_i}-U_i$ is \sq-cycle-free, $F_i$
	is a forest description of a member of $F(G_{t_i}-U_i,Y)$, $\alpha_i$ is a valid
	$\mathbb{F}_2$-labeling of $F_i$ and $d[t_i,Y,F_i,\alpha_i]=c(U_i)$. 
	We set $U = U_1 \cup U_2$, then $G_t-U = G_{t_1}-U_1 \cup G_{t_2} - U_2$ and since
	$X_t=X_{t_1}=X_{t_2}$, we have $Y=X_t \setminus U$.
	By Claim \ref{claim:join-reject}, $G_t-U$ is \sq-cycle-free. By Claim \ref{claim:join-forest},
	$F$ is a forest description of a member of $F(G_t-U,Y)$ and $\alpha$ is a valid
	$\mathbb{F}_2$-labeling of $F$.\qedhere
\end{enumerate}
\end{proof}

\begin{lemma}\label{lem:ect-complete}
	For every node $t \in V(T)$, every $U \subseteq V(G_t)$, if $G_t-U$ is \sq-cycle-free then there
	exists $F,\alpha$ such that $F$ is a forest description of a member of $F(G_t - U,X_t \setminus
	U)$, $\alpha$ is a valid $\mathbb{F}_2$-labeling of $F$, $(t,X_t \setminus U,F,\alpha)$ is
	reachable, and $d[t,X_t \setminus U,F,\alpha] \leq c(U)$. 
\end{lemma}

\begin{proof}
	By induction on $T$.
\begin{enumerate}
	\item If $t$ is a leaf node, then for $U\subseteq V(G_t)$, we have $U=\varnothing$ and $G_t-U$
	is the empty graph which is \sq-cycle-free. we have a reachable state
	$(t,\varnothing,\varnothing,\varnothing)$ with $d[t,\varnothing,\varnothing,\varnothing]= 0 \leq
	c(U)$.

	\item If $t$ is an introduce vertex node with child $t'$ introducing vertex $v$, then for
	$U\subseteq V(G_t)$ such that $G_t-U$ is \sq-cycle-free, we set $U'=U \setminus \{v\} = U \cap
	V(G_{t'})$. $G_{t'}-U'$ is an induced subgraph of $G_t-U$, hence it is \sq-cycle-free. By
	induction hypothesis, there exist $F',\alpha'$ such that $(t',X_{t'} \setminus U', F',\alpha')$
	is reachable and $d[t',X_{t'} \setminus U',F',\alpha'] \leq c(U')$.
	There are two cases, either $v \in U$ or $v\notin U$ and a transition for each case so there
	exist $F,\alpha$ such that $(t,X_t \setminus U,F,\alpha)$ is reachable and $d[t,X_t \setminus
	U,F,\alpha] \leq c(U)$.
	
	\item If $t$ is a forget vertex node with child $t'$ forgetting vertex $v$, then for $U
	\subseteq V(G_t)$ such that $G_t-U$ is \sq-cycle-free, $G_{t'}-U$ is a subgraph of $G_t-U$
	so it is \sq-cycle-free. By induction hypothesis, there exist $F',\alpha'$ such that
	$(t',X_{t'}\setminus U,F',\alpha')$ is reachable and $d[t',X_{t'}\setminus U,F',\alpha'] \leq
	c(U)$.
	By Claim \ref{claim:join-reject}, there is a transition producing $F,\alpha$ such that $(t,X_t
	\setminus U,F,\alpha)$ is reachable and $d[t,X_t \setminus U,F,\alpha] \leq c(U)$.

	\item If $t$ is a join node with children $t_1$ and $t_2$, then for $U \subseteq V(G_t)$ such
	that $G_t-U$ is \sq-cycle-free, we set $U_1=U \cap V(G_{t_1})$ and $U_2=U \cap V(G_{t_2})$.
	Hence, $G_{t_1}-U_1$ and $G_{t_2}-U_2$ are induced subgraphs of $G_t-U$ so they are
	\sq-cycle-free. By induction hypothesis, for $i \in \{1,2\}$, there exist $F_i,\alpha_i$ such
	that $(t_i,X_{t_i} \setminus U_i, F_i,\alpha_i)$ is reachable and $d[t_i,X_{t_i} \setminus U_i,
	F_i,\alpha_i]\leq c(U_i)$.
	By Claim \ref{claim:join-reject}, there is a transition producing $F,\alpha$ such that $(t,X_t
	\setminus U,F,\alpha)$ is reachable and $d[t,X_t \setminus U,F,\alpha] \leq
	d[t_1,X_{t_1}\setminus U_1,F_1,\alpha_1] + d[t_2,X_{t_2}\setminus U_2,F_2,\alpha_2] \leq c(U_1)
	+ c(U_2)=c(U)$.
\end{enumerate}
\end{proof}

\begin{lemma}\label{lem:ect-val}
	The final value of $d[r,\varnothing,\varnothing,\varnothing]$ is the weight of an optimal
	transversal.
\end{lemma}

\begin{proof}
	By applying Lemma \ref{lem:ect-complete} with $U$ a \sq-cycle transversal, because $V(G)$ is
	always a transversal, we conclude that state $(r,\varnothing,\varnothing,\varnothing)$ is
	reachable, $d[r,\varnothing,\varnothing,\varnothing] \leq c(U)$.
	By applying Lemma \ref{lem:ect-sound} to reachable state
	$(r,\varnothing,\varnothing,\varnothing)$, there exists $U^*$ a \sq-cycle transversal
	such that $d[r,\varnothing,\varnothing,\varnothing]=c(U^*)$.
	It has optimal weight due to the previous inequality.
\end{proof}

\begin{theorem}
	\textsc{Subset Feedback Vertex Set}, \textsc{Subset Odd Cycle Transversal}, and \textsc{Subset
	Even Cycle Transversal}, even in the weighted setting, can be solved in time $2^{\Oh(k \log k)}
	\cdot n$ on $n$-vertex graphs of treewidth $k$.
\end{theorem}

\begin{proof}
	We use an approximation algorithm to compute a tree decomposition of width $\Oh(k)$ in time
	$2^{\Oh(k)} \cdot n$.
	We have $2^{\Oh(k \log k)}\cdot n$ states and transitions. Since transitions are computed in time
	$k^{\mathcal{O}(1)}$, the values of all states are computed in time $2^{\Oh(k \log k)}\cdot n$.
	The solution to the problem instance is correctly computed, by Lemma \ref{lem:ect-val}.
\end{proof}

\section{Subset Feedback Vertex Set in graphs of bounded cliquewidth}\label{sec:sfvs}
\def\cw{\mathbf{cw}}
\def\poly{poly}
\def\degg{d}
\def\SFVS{\textsc{Subset Feedback Vertex Set}}
\def\wt#1{{\widetilde{#1}}}
\def\wh#1{{\widehat{#1}}}
\def\P{\mathcal{P}}

We describe a dynamic programming algorithm to solve \SFVS{} on clique-width expressions.
With a bottom-up computation, it builds small labeled forests that describe the graphs that can be
obtained by vertex deletion.

A state of our dynamic programming will consist of a node of the $k$-expression, a partially labeled forest, and a label state assignment $\mathcal{P}:[k] \rightarrow \mathcal{Q}$, with $\mathcal{Q} = \{Q_\varnothing, Q_1, Q_1^*, Q_2, Q_w, Q_w^*,Q_f\}$ the set of label states. 
State $Q_\varnothing$ is assigned to labels that are completely contained in the current deletion set.
States $Q_1$ and $Q_1^*$ are assigned to labels consisting of a single non-$S$-vertex, or a single $S$-vertex, respectively. 
States $Q_w$ and $Q_w^*$ are called \emph{waiting states}: they are assigned to labels for which we have guessed that they will be joined (only once) to a non-$S$-vertex from a label in state $Q_1$, or to an $S$-vertex from a label in state $Q_1^*$, respectively. 
State $Q_2$ is assigned to labels having at least two vertices 
not in $S$: it is assigned to labels for which we have guessed that they will be joined (potentially
several times) to either a vertex from a label in state $Q_1$, or to vertices from a label in state $Q_2$.
These guessing tricks can be seen as a form of what is
called ``expectation from the outside'' in \cite{Bui-XuanSTV13}.
We point that guessing these joins implies that labels in states $Q_w,Q_w^*,Q_2$ will eventually be connected---this is detailed below.
At last, state $Q_f$ is called \emph{final state}: it will contain vertices that will not be joined anymore, and hence that may be unlabeled. 
To summarize, states in $\mathcal{Q}$ express the following constraints on joins:
\begin{itemize}
	\item joins with a label in state $Q_\varnothing$ will be ignored; 
	\item no join with a label in state $Q_f$ will be performed;
	\item labels in state $Q_w$ (resp.~$Q_w^*$) will only be joined with those in state $Q_1$
	(resp.~$Q_1^*$); and
	\item labels in state $Q_2$ will never be joined with those in state $Q_1^*$.
\end{itemize}
 
Now, considering an $S$-cycle-free graph $\wt{G}$ obtained by vertex deletion, we will say that a label $i$ is \emph{compatible} with label state:
\begin{itemize}
	\item $Q_\varnothing$ if no vertex of $\wt{G}$ is labeled $i$;
	\item $Q_1$ if exactly one vertex of $\wt{G}$ is labeled $i$, and it is not in $S$;
	\item $Q_1^*$ if exactly one vertex of $\wt{G}$ is labeled $i$, and it is in $S$;
	\item $Q_2$ if at least two vertices of $\wt{G}$ are labeled $i$, they are not in $S$, and
	no $S$-path in $\wt{G}$ has both its endpoints labeled $i$;
	\item $Q_w$ if at least two vertices of $\wt{G}$ are labeled $i$, at least one $S$-vertex is labeled $i$, and no $S$-path in $\wt{G}$  has both its endpoints labeled $i$;
	\item $Q_w^*$ if at least two vertices of $\wt{G}$ are labeled $i$, no path in $\wt{G}$ has both its
	endpoints labeled $i$; and
	\item $Q_f$ if at least two vertices of $\wt{G}$ are labeled $i$.
\end{itemize}
These conditions, together with the constraints on joins that are expressed above, aim to capture cases for which a join between labels of pairs of label states will not create $S$-cycles---this will be explicited in proofs and illustrated in Figure \ref{figjoin}.
In the following, we say that a label state assignment $\mathcal{P}$ is \emph{compatible} with $\wt{G}$ if each label is compatible
with its state in this graph. 
Note that looking at the properties of vertices in a label in part
gives the label state assignment that it should have: the conflicts are for choosing between $Q_f$,
$Q_w^*$ and, based on the presence or not of an $S$-vertex, either $Q_w$ or $Q_2$. 
This is expected because these states contain the information on a guess on what will later be added to the graph.

\begin{figure}[!h]
 \begin{subfigure}{0.29\linewidth}
  \centering{\includegraphics{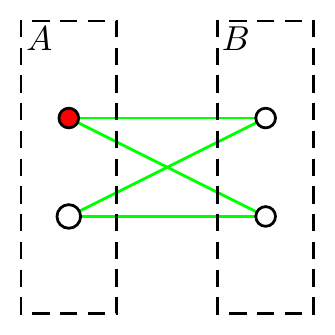}}
  \caption{Case 1: $|A|\geq 2$, $|B|\geq 2$,\\ and $|A\cap S|+|B\cap S|\geq 1$.}
  \label{fig_join_a}
 \end{subfigure}
 \begin{subfigure}{0.39\linewidth}
  \centering{\includegraphics{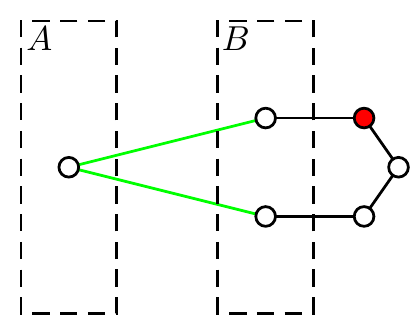}}
  \caption{Case 2: $|A|\geq 1$, $|B|\geq 2$, and there\\ is an $S$-path with endpoints in $B$.}
  \label{fig_join_b}
 \end{subfigure}
 \begin{subfigure}{0.29\linewidth}
  \centering{\includegraphics{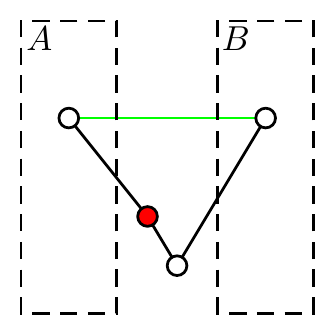}}
  \caption{Case 3: $|A|\geq 1$, $|B|\geq 1$,\\ and there is an $S$-path connecting a vertex of $A$ to one of $B$.}
  \label{fig_join_c}
 \end{subfigure}
  \caption{The three cases when a join (depicted in green) creates an $S$-cycle. The figures illustrates the smallest number of vertices of $S$ required. Thus, up to symmetry, the vertices depicted in red have to be in $S$ while the vertices in white may or may not be in $S$.}
 \label{figjoin}
\end{figure}

Let us now introduce an auxiliary partially labeled graph which will conveniently represent the connectedness implied by guesses we made so far when assigning labels to label states, while simplifying the manipulation of labels.
We point that this auxiliary graph will \emph{not} be computed by the algorithm: it shall only be used in the proofs.
Given a labeled graph $\wt{G}$ and a label state assignment $\mathcal{P}$, we denote by $H(\wt{G},\mathcal{P})$ the partially labeled graph obtained from $\wt{G}$, by conducting the following modifications for each label $i$:
\begin{itemize}
	\item if $i$ is in state $Q_2$ or $Q_w$, we add a vertex labeled $i$, connect it to other vertices labeled $i$, and unlabel these vertices, making the added vertex the only vertex labeled $i$;
	\item if $i$ is in state $Q_w^*$, we add an $S$-vertex labeled $i$, connect it to other vertices
	labeled $i$, and unlabel these vertices, making the added vertex the only vertex labeled $i$; and
	\item if $i$ is in state $Q_f$, we unlabel vertices labeled $i$.
\end{itemize}
Note that in the auxiliary graph, we add vertices that are not part of the original graph.
The role of these vertices---for states $Q_2$, $Q_w$, and $Q_w^*$---is to represent the label $i$ as if it was connected (which will eventually be the case as we guessed a later join), as well as manipulating nonempty labels as single vertices: for $\wt{G}$ compatible with $\mathcal{P}$, each nonempty label $i$ contains exactly one vertex in $H(\wt{G},\mathcal{P})$, which we call \emph{representative} of $i$ in $H(\wt{G},\mathcal{P})$, and that we denote by $h(i)$.

Recall that, when $\mathcal{P}$ is compatible with $\wt{G}$, some connectedness conditions are
satisfied by label states.
We say that a partially labeled multigraph $\widehat{F}$ \emph{expresses the connectedness} in $H(\wt{G},\mathcal{P})$, for $\wt{G}$ and $\mathcal{P}$ compatible, if:
\begin{itemize}
	\item for each label $i$, there is at most one vertex labeled $i$ in $\widehat{F}$;
	\item to every vertex $h(i)$ in $H(\wt{G},\mathcal{P})$ corresponds a vertex $r(i)$
	labeled $i$ in $\widehat{F}$: we call it the \emph{representative} of label $i$ in $\wh{F}$, and $r(i)$ is an $S$-vertex if and only if $h(i)$ is an
	$S$-vertex; and
	\item for any two vertices $h(i),h(j)$ in $H(\wt{G},\mathcal{P})$, there exists a
	$h(i)$--$h(j)$ path in $H(\wt{G},\mathcal{P})$ if and only if there exists a $r(i)$--$r(j)$ path in $\widehat{F}$, and there
	exists a $h(i)$--$h(j)$ $S$-path in $H(\wt{G},\mathcal{P})$ if and only if there exists a $r(i)$--$r(j)$ $S$-path in $\widehat{F}$.
\end{itemize}

We are now ready to introduce reduction rules which, when applied on the multigraph $\widehat{F}$
expressing the connectedness in $H(\wt{G},\mathcal{P})$, will produce the aforementioned partially labeled forest.
The idea behind this forest is that, to check the existence of ($S$-)paths linking representatives of labels $i$ and $j$, unlabeled vertices of degree at most two in such ($S$-)paths may be ``contracted'' as long as we do not remove all ($S$-)vertices on these paths. 
In the following for a partially labeled multigraph $\widehat{F}$, we denote by $\Red(\widehat{F})$
the forest obtained from $\widehat{F}$ by applying the following reduction rules:
\begin{itemize}
	\item for each nontrivial $2$-connected component $C$, we introduce an unlabeled vertex, call it \emph{central vertex} of $C$, connect it to 
	vertices of $C$, and remove all other edges inside the component;
	\item we iteratively remove unlabeled vertices of degree at most one;
	\item for each maximal $S$-path with internal unlabeled vertices of degree two, we replace it by
	connecting the endpoints to a single new unlabeled $S$-vertex; and
	\item for each maximal path with internal unlabeled vertices of degree two that is not an $S$-path, we replace it by a
	single edge between its endpoints.
\end{itemize}

It is easily seen that the produced graph is indeed a forest as the graph of nontrivial
$2$-connected components of any graph is a tree, and each nontrivial
$2$-connected component is replaced by a star.

\begin{claim}\label{claim:sfvs-forest-size}
	$\Red(\widehat{F})$ has $\Oh(k)$ vertices.
\end{claim}

The proof of Claim~\ref{claim:sfvs-forest-size} follows the lines of the one of Lemma \ref{lem:red-forest-size}: it is omitted here.

\begin{lemma}\label{lem:red-con}
	Let $\wt{G}$ be a labeled graph, and $\mathcal{P}$ be a label state assignment compatible with $\wt{G}$. 
	If $H(\wt{G},\mathcal{P})$ is $S$-cycle-free and $\widehat{F}$ expresses the connectedness in
	$H(\wt{G},\mathcal{P})$, then $F=\Red(\widehat{F})$ expresses the connectedness in
	$H(\wt{G},\mathcal{P})$.
\end{lemma}

\begin{proof}
	Note that according to the reduction rules, no labeled vertex is removed nor added to $\widehat{F}$ when computing $F=\Red(\widehat{F})$.
	Hence the first two conditions on expressing the connectedness in $H(\wt{G},\mathcal{P})$ are fulfilled by $F$, whenever they are by $\widehat{F}$.
	It remains to show that the existence of ($S$-)paths between labeled vertices of $\widehat{F}$ is preserved in $F$.
	First, note that any path $P$ of $\widehat{F}$ going through a nontrivial $2$-connected
	component $C$ can be turned into a $3$-vertex path of $F$ going through the central vertex of the component, and having the first and last vertices of $P$ restricted to $C$ as endpoints.
	Hence the first reduction rule preserves the desired property.
	Second, no unlabeled vertex of degree at most one lies on a path linking two labeled vertices.
	Hence the second reduction rule preserves the desired property, and the same conclusion is straightforward for the last two reduction rules.
\end{proof}

A \emph{state} of the dynamic programming algorithm is a tuple $(t,F,\mathcal{P})$, where $t \in V(T)$, 
$F$ is a partially labeled forest, and $\mathcal{P} : [k] \to \mathcal{Q}$ is a
label state assignment. 
We say that $(t,F,\mathcal{P})$ is \emph{admissible} if there exists $X\subseteq V(G)$ such that
$\mathcal{P}$ is compatible with $G_t-X$, $H(G_t-X,\mathcal{P})$ is $S$-cycle-free, and $F$
expresses the connectedness in $H(G_t-X,\mathcal{P})$.
Our dynamic programming algorithm will not consider all possible states, but compute a value $d[t,F,\mathcal{P}]$ for some states $(t,F,\mathcal{P})$. 
We call \emph{reachable} a state that is considered by the algorithm.
We will show that reachable states are admissible, that for every $t \in V(T)$, for
each $X \subseteq V(G_t)$, if $G_t-X$ is $S$-cycle-free, then there exists a reachable state
$(t,F,\mathcal{P})$ such that $d[t,F,\mathcal{P}]\leq |X|$, and that the optimal value for SFVS on
the given instance is the minimum of values $d[r,F,\mathcal{P}]$ where $r$ is the root of the
$k$-expression.

First, let us slightly modify our clique-width expression in order to simplify the description of our
computations. We double the set of labels, denoting them by $\{1,...,k,1',...,k'\}$, and replace
each disjoint union node $t$ with children $t_1,t_2$ by the following subexpression:
$\rho_{1' \rightarrow 1}(\dots\rho_{k' \rightarrow k}(G_{t_1} \oplus (\rho_{1 \rightarrow
1'}(\dots\rho_{k \rightarrow k'}(G_{t_2})))))$. This gives the property that in disjoint union
nodes, each label is used by at most one of the children nodes.

We now describe the bottom-up computation of reachable states for each possible type of node in the
clique-width expression.

\begin{description}
	\item[Leaf node.] If $t$ is a leaf node with $G_t=i(v)$, two cases arise. 
	Either $v$ is deleted which is described by state $(t,F_\varnothing,\mathcal{P}_\varnothing)$
	initialized with value $c(v)$, where $F_\varnothing$ is the empty graph, and $\mathcal{P}_\varnothing$ is the function that maps every $i\in [k]$ to $Q_\varnothing$.
	Otherwise we keep $v$, which is described by state $(t,F,\mathcal{P})$
	where $F$ consists of the isolated vertex $v$, $\mathcal{P}(i)=Q_1^*$ if $v \in S$,
	$\mathcal{P}(i)=Q_1$ otherwise, and, for all $j \neq i$, $\mathcal{P}(j)=Q_\varnothing$.

	\item[Join node.] 
	Let $t$ be a join node with $G_t=\eta_{i \times j}(G_{t'})$. 
	For each reachable state $(t',F',\mathcal{P}')$, we proceed as follows.
	If the representatives of $i$ and $j$ are connected by an $S$-path in $F'$, we do nothing. 
	Otherwise, we will construct states
	$(t,F,\mathcal{P})$ defined in the following cases, depending on $\mathcal{P}'(i)$ and
	$\mathcal{P}'(j)$, starting with $F:=F'$ and $\mathcal{P}:=\mathcal{P}'$:
	\begin{itemize}
		\item if one of $i$ and $j$ is in state $Q_\varnothing$, we do not modify $F$ nor $\P$;
		\item if $i$ and $j$ are in states $Q_1$ or $Q_1^*$, we add an edge between the
		representatives of $i$ and $j$ in $F$;
		\item if $i$ and $j$ are in states $Q_1$ or $Q_2$, we add an edge between the
		representatives of $i$ and $j$ in $F$, and if $i$ or $j$ are in state $Q_2$ they are allowed
		to change to $Q_f$ in $\mathcal{P}$, if they do we also unlabel their representative: we
		enumerate all possibilities here;
		\item if $i$ and $j$ are in states $Q_1^*$ and $Q_w^*$, we identify their representative in $F$:
		the resulting vertex has its label in state $Q_1^*$, and the label in state $Q_w^*$
		is assigned state $Q_f$ in $\mathcal{P}$; and
		\item if $i$ and $j$ are in states $Q_1$ and $Q_w$, we identify their representative in $F$: the
		resulting vertex has its label in state $Q_1$, and the label in state $Q_w$ is
		assigned state $Q_f$ in $\mathcal{P}$.
	\end{itemize}
	For each such cases, we reduce $F$ and propagate the
	value $d[t',F',\mathcal{P}']$ to the states $(t,F,\mathcal{P})$, where $\mathcal{P}$ is the
	modified label state assignment.

	\item[Renaming label node.] Let $t$ be a renaming label node with $G_t=\rho_{i \rightarrow j}(G_{t'})$.
	For each reachable state $(t',F',\mathcal{P}')$, 
	we construct states $(t,F,\mathcal{P})$ starting with $\mathcal{P}:=\mathcal{P}$ and $F:=F'$ by
	first setting $\mathcal{P}(i)=Q_\varnothing$, and proceeding as follows depending on
	$\mathcal{P}'(i)$ and $\mathcal{P}'(j)$:
	\begin{itemize}
			\item if $i$ and $j$ are in a state among $\{Q_f,Q_1,Q_1^*\}$, we unlabel the representatives of $i$ and $j$ in $F'$, and set $\mathcal{P}(j)=Q_f$;
			\item if one of $i$ and $j$ is in state $Q_\varnothing$, then either $i$ is in state
			$Q_\varnothing$ and we do nothing, or $j$ is in state $Q_\varnothing$, we assign it to
			the other label state, and the vertex of $F$ labeled $i$ is relabeled $j$;
			\item if $i$ and $j$ are in state $Q_1$, and the representatives of $i$ and $j$ are not
			connected by a path in $F'$, in $F$, we add an $S$-vertex labeled $j$, connect it to these
			vertices, and unlabel them. 
			Label $j$ is then assigned state $Q_w^*$ in $\mathcal{P}$;
			\item if one of $i$ and $j$ is in state $Q_1^*$, the other is in state $Q_1$ or $Q_1^*$,
			and the representatives of $i$ and $j$ are not connected by a path in $F'$, we consider
			two possibilities depending on whether they will be joined to a vertex of $S$, or to a
			vertex of $V(G) \setminus S$. 
			First, in $F$, we add a new vertex labeled $j$, connect it to the representatives of $i$
			and $j$, and unlabel the representatives of $i$ and $j$. 
			Then, if the new vertex is chosen to be in $S$, $j$ is assigned state $Q_w^*$ in
			$\mathcal{P}$.
			Otherwise, $j$ is assigned state $Q_w$ in $\mathcal{P}$;
			\item if $i$ and $j$ are in states $Q_\alpha$ and $Q_\beta$, for $\alpha,\beta \in \{1,2,w\}$,
			and the representatives of $i$ and $j$ are not connected by an $S$-path in $F'$, in $F$, we
			identify the representatives of $i$ and $j$: the resulting vertex is of label $j$, and
			$j$ is assigned state $Q_\delta$ in $\mathcal{P}$ with $\delta:=\max \limits_{1<2<w}
			\{2,\alpha,\beta\}$;
			\item if one of $i$ and $j$ is in state $Q_1^*$, the other is in state $Q_w$ or $Q_2$, and the
			representatives of $i$ and $j$ are not connected by a path in $F'$, in $F$, we add an
			edge between the representatives of $i$ and $j$, and the representative of the label in
			state $Q_1^*$ becomes unlabeled, while the other vertex is given label $j$. Label $j$ is
			assigned state $Q_w$ in $\mathcal{P}$; and
			\item if one of $i$ and $j$ is in state $Q_w^*$, the other is in state $Q_1$ or $Q_1^*$,
			and the representatives of $i$ and $j$ are not connected by a path in $F'$, in $F$, we add an
			edge between the representatives of $i$ and $j$, and the representative of the label
			in state $Q_1$ or $Q_1^*$ becomes unlabeled, while the other vertex is given label $j$.
			Label $j$ is assigned state $Q_w^*$ in $\mathcal{P}$.
		\end{itemize}
		For each such cases, we reduce $F$, and we propagate the value $d[t',F',\mathcal{P}']$ to
		the state $(t,F,\mathcal{P})$, where $\mathcal{P}$ is the modified label state assignment.

	\item[Disjoint union node.] If $t$ is a disjoint union node with $G_t=G_{t_1} \oplus G_{t_2}$,
		for each pair of reachable states $(t_1,F_1,\mathcal{P}_1)$, $(t_2,F_2,\mathcal{P}_2)$,
		since they use disjoint sets of labels, we can simply define $F = F_1 \oplus F_2$. 
		The label state assignment $\mathcal{P}$ is defined by $\mathcal{P}(i)=\mathcal{P}_1(i)$ for
		$i \in [k]$ and $\mathcal{P}(i')=\mathcal{P}_2(i')$ for $i' \in \{1',\dots,k'\}$.
		The value $d[t_1,F_1,\mathcal{P}_1]+d[t_2,F_2,\mathcal{P}_2]$ is propagated to
		state $(t,F,\mathcal{P})$.
\end{description}

We now prove the correctness of the algorithm.

\begin{lemma}\label{lem:adm}
	For each reachable state $(t,F,\mathcal{P})$, there exists $X \subseteq V(G_t)$ such that
	$\mathcal{P}$ is compatible with $G_t-X$, $H(G_t-X,\mathcal{P})$ is $S$-cycle-free, $F$ expresses the
	connectedness in $H(G_t-X,\mathcal{P})$, and ${d[t,F,\mathcal{P}]=c(X)}$.
\end{lemma}

\begin{proof}
	We proceed by induction on $T$. 
	For convenience in the following, let us denote by $H$ and $H'$ the graphs $H(G_t-X,\mathcal{P})$ and
	$H(G_{t'}-X,\mathcal{P}')$, by $h(i)$ and $h'(i)$ the representative of label $i$ in $H$
	and~$H'$, and by $r(i)$ and $r'(i)$ the representative of label $i$ in $F$ and $F'$,
	respectively. 
	By an abuse of notation, we will denote by $F$ both the graph expressing the connectedness in
	$H$ and its reduced forest $\Red(\widehat{F})$; we may do so as a reduction of $F$ is computed
	at each transition of our algorithm, and by Lemma \ref{lem:red-con}, the obtained forest
	expresses the connectedness in $H$.
	\begin{enumerate}	
		\item If $t$ is a leaf node, the property is trivial.

		\item Let $t$ be a join node with $G_t=\eta_{i \times j}(G_{t'})$ and $(t,F,\mathcal{P})$ be
		a reachable state.
		Then there exists a reachable state $(t',F',\mathcal{P}')$ which gave the
		optimal value to $(t,F,\mathcal{P})$, i.e., such that $d[t,F,\mathcal{P}] = d[t',F',\mathcal{P}']$. 
		By induction hypothesis, there exists $X$ such that $\mathcal{P}'$ is
		compatible with $G_{t'}-X$, $H'$ is $S$-cycle-free, $F'$ expresses the connectedness in $H'$, and
		$d[t',F',\mathcal{P}'] = c(X)$. 
		Hence $d[t,F,\mathcal{P}]=c(X)$.
		It remains to show that the properties for $\mathcal{P}$, $H(G_t-X,\mathcal{P})$, and $F$
		hold for each of the transitions, depending on the label states assigned to $i$ and $j$ in
		$\mathcal{P'}$.
		The following observation will be convenient.

		\begin{observation}\label{obs:H-G-link}
			Let $\wt{G}$ be a graph, $\P$ be a label state assignment, $i$ be a label such that $\P(i)\in \{Q_2,Q_w\}$, and $j$ be a label such that $\P(j)=Q_w^*$.
			Then there is an $S$-cycle containing $h(i)$ in $H(\wt{G},\mathcal{P})$ whenever there is a $S$-path having elements of $i$ as its endpoints in $\wt{G}$, and then there is an $S$-cycle containing $h(j)$ in $H(\wt{G},\mathcal{P})$ whenever there is a path having elements of $j$ as its endpoints in $\wt{G}$.
		\end{observation}

		\begin{itemize}
			\item If one label is in state $Q_\varnothing$, by compatibility of $\mathcal{P}'$ and
			$G_{t'}-X$, we know that the join added no edge to the graph, i.e., $G_t-X = G_{t'}-X$. 
			By the described transitions, $\mathcal{P}=\mathcal{P}'$ and $F=F'$.
			Hence $H=H'$, and the conclusion follows.

			\item We now deal with the two next cases: these cases do not involve waiting states. 
			From the transition we know that there is no $S$-path between $i$ and $j$ in $F'$, and
			since $F'$ expresses the connectedness in $H'$, there is no $S$-path between
			representatives of $i$ and $j$ in $H'$. 
			Let us show that $H$ is $S$-cycle-free.
			Consider the graph $H'+A$, where $A=E(G_t-X) \setminus E(G_{t'}-X)$ is the set of edges
			added by the join.
			We note that, according to the transitions, either $\P=\P'$, or $\P'$ differs from $\P$
			by assigning one of $i$ and $j$ to $Q_f$.
			In both cases, $H$ is a subgraph of $H'+A$.
			Recall that $H'$ is $S$-cycle-free, and no $S$-path connects the representatives of $i$
			and $j$.
			Hence if $H'+A$ contains an $S$-cycle, it must be contained in $A$, which is
			excluded by the compatibility of $\P'$ and $G_{t'}-X$, more particularly by the
			constraints induced by compatibility on the two joined labels: a join between labels of
			$Q_1$ and $Q_1^*$, or labels of $Q_1$ and $Q_2$, cannot create a cycle.  Hence $H$ is
			$S$-cycle-free.
			The claim that $G_t-X$ and $\P$ are compatible follows from Observation~\ref{obs:H-G-link} for the connectivity constraints, and on the fact that either $\P=\P'$ or $\P$ assigns one of $i$ and $j$ to $Q_f$, for the vertex constraints.
			Finally, note that any path having its endpoints in $i$ and $j$ after the join has its ``connectedness'' expressed by $F$ after adding the edge $r(i)r(j)$.
			This concludes that case.

			\item The other transitions involve a label in state $Q_1$ or $Q_1^*$, joined with a
			label in a waiting state $Q_w$ or $Q_w^*$, respectively. 
			W.l.o.g., let us assume that $i$ is in
			state $Q_1$ or $Q_1^*$ and $j$ in state $Q_w$ or $Q_w^*$.
			Let $v=h'(i)$ and $w=h'(j)$.
			Consider the graph $H'_{/vw}$, and denote by $u$ the vertex obtained by identification
			of $v$ and $w$, where $u$ has label $i$.
			We show that $H=H'_{/vw}$.
			By the transitions, $i$ stays in his state, and $j$ is assigned state $Q_f$ in $\mathcal{P}$. 
			Hence in $H$, label $i$ consists of the single element $v$, and label $j$ has no
			representative.
			Thus $V_i(H)=V_i(H'_{/vw})$ and $V_j(H)=V_j(H'_{/vw})$. 
			Now, the performed join on $G_t'-X$ adds every edge between $v$ and $V_j(G_{t'})$.
			Since $N_{H'}(w)=V_j(G_{t'})$, these added edges are exactly those incident to $u$ in $H'_{/vw}$ and that were not already present in $H'$.
			We conclude that $H=H'_{/vw}$ as desired.
			Since $H'$ is $S$-cycle-free, an $S$-cycle of $H$ must contain one of the identified vertices.
			However these vertices are not connected by an $S$-path, by the first condition expressed in the decription of the transitions. 
			Consequently $H$ is $S$-cycle-free, and the compatibility of $G_t-X$ and $\P$ then follows from Observation~\ref{obs:H-G-link} for the connectivity constraints, and from the fact that $\P$ only differs on $\P'$ on the fact that it assigns $j$ to $Q_f$, for the vertex constraints.
			Clearly, identifying the the representatives in $F'$ preserves expressing the
			connectedness in $H'_{/vw}$, and the conclusion follows.
		\end{itemize}		

		\item Let $t$ be a relabel node with $G_t=\rho_{i \rightarrow j}(G_{t'})$, and $(t,F,\mathcal{P})$ be a reachable state.
		Then there exists a reachable state $(t',F',\mathcal{P}')$ which gave the
		optimal value to $(t,F,\mathcal{P})$, i.e., such that $d[t,F,\mathcal{P}]=d[t',F',\mathcal{P}']$. 
		By induction hypothesis, there exists $X$ such that $\mathcal{P}'$ is compatible with
		$G_{t'}-X$, $H'$ is $S$-cycle-free, $F'$ expresses the connectedness in $H'$, and
		$d[t',F',\mathcal{P}']=c(X)$. 
		Hence $d[t,F,\mathcal{P}]=c(X)$.
		We show that the properties for $\mathcal{P}$, $H(G_t-X,\mathcal{P})$, and $F$ hold for each of the transitions, depending on the label types of $i$ and $j$.
		We stress that in the following, $G_t-X$ and $G_{t'}-X$ are isomorphic: they only differ by their label assignments.
		
		\begin{itemize}
			\item We consider the first case where $i$ and $j$ are in a state among $\{Q_f,Q_1,Q_1^*\}$.
			Since $\mathcal{P}'$ is compatible with $G_{t'}-X$, and the transition assigns $i$ and $j$ to $Q_\varnothing$ and $Q_f$, it is easily seen that $\mathcal{P}$ is compatible with $G_t-X$: after the relabeling, $V_i(G_t-X)$ is empty, and $|V_j(G_t-X)|\geq 2$.
			Furthermore, as labels $i$ and $j$ have no representative in $H$, and the representative of $i$ and $j$ in $H'$, if they exist, are vertices of $G_t-X$, we conclude that $H$ and $H'$ are isomorphic. 
			Consequently $H$ is $S$-cycle-free.
			At last, as $F$ only differs with $F'$ on $i$ and $j$ being unlabeled, and $H$ has no representative for labels $i$ and $j$, $F$ expresses the connectedness in $H$.

			\item The transition with one of $i$ and $j$ being in state $Q_\varnothing$ either corresponds to doing nothing, or to
			swapping labels $i$ and $j$ in $G_t-X$ and $F$: the conclusion follows.

			\item We now consider the transition with both $i$ and $j$ in state $Q_1$.
			Since their representatives in $F'$ are not connected by a path, their representatives in
			$H'$ are not connected either.
			Since $G_{t'}-X$ and $\P'$ are compatible, labels $i$ and $j$ are nonempty in
			$G_{t'}-X$, and hence $|V_j(G_t-X)|\geq 2$.
			Consequently, $j$ is compatible with state $Q_w^*$ in $G_t-X$ after relabeling. 
			As $V_i(G_t-X)$ is empty, we conclude that $\P$ and $G_t-X$ are compatible.
			By construction, $H$ contains one representative for label $j$, and it is an $S$-vertex.
			The same holds for $F$ from the transition.
			Since $F'$ expresses the connectedness in $H'$, and given that $h(j)$ and $r(j)$ are
			connected to $V_j(G_t-X)$, and to $r'(i)$ and $r'(j)$, respectively, it is easily
			checked that for any two representatives $h(\alpha),h(\beta)$ in $H$, there is a
			$h(\alpha)$--$h(\beta)$ path going through $h(j)$ in $H$ if and only if there is a
			$r(\alpha)$--$r(\beta)$ path going through $r(j)$ in $F$.
			We deduce that $F$ expresses the connectedness in $H$.
			At last, $H$ is $S$-cycle-free because $H'$ is $S$-cycle-free, and no path
			of $H'$ links two vertices labeled $j$ in $G_{t'}-X$.

			\item We now consider the transition with a label in state $Q_1^*$, and the other in state $Q_1$ or $Q_1^*$. 
			Since the representatives of $i$ and $j$ are not connected by a path in $F'$, and since $F'$ expresses the connectedness in $H'$, they are not connected in $G_t-X$.
			As $V_i(G_{t'}-X)$ and $V_j(G_{t'}-X)$ are nonempty, we deduce that $j$ is compatible with state $Q_w^*$ or $Q_w$, depending on whether $j$ is a $S$-vertex or not. 
			As $V_i(G_t-X)$ is empty, we conclude that $G_t-X$ and $\P$ are compatible.
			The fact that $F$ expresses the connectedness in $H$ follows from the same arguments as in the previous case: the new representative of $j$ in $F$
			corresponds exactly to the new representative of $j$ in $H$. 
			Similarly, we deduce that $H$ is $S$-cycle-free because $H'$ is
			$S$-cycle-free, and no path of $H'$ links two vertices labeled $j$ in $G_{t'}-X$.

			\item We now consider the transition with labels in states $Q_\alpha$ and $Q_\beta$, for
			$\alpha,\beta \in \{1,2,w\}$.
			From the transition, we also know that the representatives of $i,j$ in $F'$ are not connected by an $S$-path.
			By the compatibility of $\mathcal{P}'$ and $G_{t'}-X$, labels $i$ and $j$ have a
			representative in $H'$, and no pair of vertices in such labels in $G_t-X$ are connected
			by an $S$-path.  Hence after relabeling, no pair of vertices labeled $j$ in $G_t-X$ is
			connected by an $S$-path.
			If the new label state is $Q_2$, then the previous states were
			forcing $j$ to contain no $S$-vertex. 
			If the new label state is $Q_w$, then one of the
			previous states was $Q_w$, which also requires $i$ to contain at least one $S$-vertex.
			Now as $V_i(G_t-X)$ is empty, $\mathcal{P}$ and $G_t-X$ are compatible.
			Note that the vertices that were put to label $j$ when relabeling are now connected in $H$.
			Identifying $r'(i)$ and $r'(j)$ in $F'$ thus produces the desired result of expressing the connectivity in $H$.
			At last, $H$ is $S$-cycle-free because the new connections induced by the relabeling are between vertices that were not connected by $S$-paths in $H'$.

			\item We now consider the transition with a label in state $Q_1^*$, and the other in state $Q_w$ or $Q_2$. 
			Representatives of $i,j$ are not connected by a path in $F'$. 
			As one label is in state $Q_1^*$, by compatibility of $\mathcal{P}'$ and $G_{t'}-X$, it
			contains an $S$-vertex.
			From $F'$ expressing the connectedness in $H'$, we conclude that $j$ is compatible with $Q_w$.
			Since $V_i(G_t-X)$ is empty, we have that $\mathcal{P}$ and $G_{t}-X$ are compatible.
			Then, the fact that $F$ expresses the connectedness in $H$ follows from the fact that new connections created by $h(j)$ in $H$ are expressed by an edge between $r'(i)$ and $r'(j)$ in $F$.
			At last, $H$ is $S$-cycle-free because $H'$ is $S$-cycle-free, and vertices labeled
			$i,j$ in $G_{t'}-X$ were not connected by a path in $H'$.

			\item We now consider the transition with a label in state $Q_w^*$, and the other in $Q_1$ or $Q_1^*$.
			We in addition know that the representatives of $i,j$ in $F'$ are not connected by a
			path, and because $F'$ expresses the connectedness in $H'$, the same holds in $H'$.
			As one label is assigned $Q_w^*$ in $\mathcal{P}'$, we deduce from compatibility that
			there is no path in $H'$ between its vertices of $G_{t'}-X$.
			Since $V_i(G_t-X)$ is empty, we get that $\mathcal{P}$ and $G_t-X$ are compatible.
			The same argument as in the previous case show that $F$ expresses the connectedness in
			$H$, and that $H$ is $S$-cycle-free: new connections created by $h(j)$ in $H$ are
			expressed by an edge between $r'(i)$ and $r'(j)$ in $F$.  
		\end{itemize}

		\item Let $t$ be a disjoint union node with $G_t=G_{t_1} \oplus G_{t_2}$, and $(t,F,\mathcal{P})$
		be a reachable state.
		Then $(t,F,\mathcal{P})$ was constructed by a transition from some $(t_1,F_1,\mathcal{P}_1)$
		and $(t_2,F_2,\mathcal{P}_2)$. By induction hypothesis applied to $(t_1,F_1,\mathcal{P}_1)$
		and $(t_2,F_2,\mathcal{P}_2)$ and the fact that the union is disjoint, there exist $X_1$ and $X_2$ such that,
		for $X:=X_1 \cup X_2$, $H(G_t-X,\mathcal{P})$ is $S$-cycle-free,
	 	$\mathcal{P}$ is compatible with $G_t-X$, $F$ expresses the
		connectedness in $H(G_t-X,\mathcal{P})$, and $d[t,F,\mathcal{P}]= c(X_1)+c(X_2)=c(X)$.
		Hence the conclusion.
		\qedhere
	\end{enumerate}
\end{proof}

In the following, we say that graph $G$ is \emph{less connected} than graph $G'$ if $E(G) \subseteq
E(G)$. For auxiliary graphs $H=H(G,\mathcal{P})$ and $H'=H(G',\mathcal{P})$ with $V(G)=V(G')$,
denote by $R$ and $R'$ their accessibility relation restricted to $V(G)=V(G')$ (paths through other
vertices are allowed), we say that $H$ is less connected than $H'$ if $R \subseteq R'$.

\begin{lemma}\label{lem:cplt}
	For each $t \in V(T)$, for each $X \subseteq V(G_t)$, and for each label state assignment
	$\mathcal{P}$ that is compatible with $G_t-X$ and such that $H(G_t-X,\mathcal{P})$ is
	$S$-cycle-free, there exists $F$ such that $F$ expresses the connectedness in $H(G_t-X,\P)$,
	$(t,F,\mathcal{P})$ is reachable and $d[t,F,\mathcal{P}] \leq |X|$.
\end{lemma}

\begin{proof}
	We proceed by induction on $T$.
	\begin{enumerate}
		\item If $t$ is a leaf node, then the statement holds by construction.

		\item Let $t$ be a join node with $G_t=\eta_{i \times j}(G_{t'})$, $X\subseteq V(G_t)$, and $\mathcal{P}$ be a label state assignment such that $\mathcal{P}$ is compatible with $G_t-X$ and $H(G_t-X,\mathcal{P})$ is $S$-cycle-free.
		First note that if label $\ell$ is compatible with $\P(\ell)$ in $G_t-X$, then it is
		compatible with $\P(\ell)$ in $G_{t'}-X$. This allows us for labels $\ell \in
		[k]\setminus\{i,j\}$ to keep label state $\P(\ell)$ in the label state assignments $\P'$ and
		$\P''$ that we will construct, while preserving compatibility. In the following, we will
		only justify the compatibility for $i$ and $j$. For the transitions with a label in state
		$Q_\varnothing$, we immediately conclude that the forest given by the transition expresses
		the connectedness in $H(G_t-X,\mathcal{P})$ from the induction hypothesis. For other
		transitions, by comparing paths that appear when constructing $F$ from $F'$ with paths
		that are in $H(G_t-X,\P)$ but not $H(G_{t'}-X,\P')$, we deduce that $F$ expresses the
		connectedness in $H(G_t-X,\P)$, this is already done in detail in the previous lemma. We
		consider the following cases based on the states of $i$ and $j$:
		
		\begin{itemize}
			\item We first consider the case where $\P(i),\P(j) \in \{Q_2,Q_w,Q_w^*,Q_f\}$. Since
			$\P$ is compatible with $G_t-X$, $|V_i(G_t-X)| \geq 2$, $|V_j(G_t-X)| \geq 2$.
			Then as $H(G_t-X,\mathcal{P})$ is $S$-cycle-free, it must be that none of the vertices
			in $V_i(G_{t'}-X)$ and $V_j(G_{t'}-X)$ are $S$-vertices, that no $S$-path connects a
			vertex of $i$ to a vertex of $j$ in $G_{t'}-X$, and that no $S$-path connects two
			vertices from $i$, or two vertices from $j$, in $G_{t'}-X$.
			Under these constraints, labels $i$ and $j$ are compatible with state $Q_2$ in $G_{t'}-X$.
			We set $\P'(i)=\P'(j)=Q_2$, $\P'$ is compatible with $G_{t'}-X$.
			Let us show that $H(G_{t'}-X,\mathcal{P'})$ is $S$-cycle-free.
			Suppose toward a contradiction that there is an $S$-cycle $C$ in $H(G_{t'}-X,\mathcal{P}')$. 
			Since $C$ is not a subgraph of $H(G_t-X,\mathcal{P})$, it
			must contain at least one of the representatives of $i$ and $j$. 
			Suppose w.l.o.g.~that $C$ contains the representative $i$. 
			Then, if $C$ contains an element of $V_j(G_{t'}-X)$, we deduce the existence of an
			$S$-path connecting $i$ and $j$ in $H(G_{t'}-X,\mathcal{P}')$, which can be completed to
			form an $S$-cycle in $H(G_t-X,\mathcal{P})$ using one of the edges added by the join
			operation. 
			This contradicts the fact that $H(G_t-X,\mathcal{P})$ is $S$-cycle-free.		
			If $C$ does not contain any vertex from $V_j(G_{t'}-X)$, replacing the representative of
			$i$ by an arbitrary vertex in $V_j(G_{t'}-X)$ also yields an $S$-cycle in
			$H(G_t-X,\mathcal{P})$, and leads to the same contradiction.
			We conclude that $H(G_{t'}-X,\mathcal{P}')$ is $S$-cycle-free, as desired.

			By induction hypothesis with, there exists $F'$ such that $F'$ expresses the
			connectedness in $H(G_t-X,\P)$, $(t',F',\mathcal{P}')$ is reachable, and
			$d[t',F',\mathcal{P}'] \leq |X|$.
			Since the join connected vertices of $V_i(G_{t'}-X)$ 
      to vertices of $V_j(G_{t'}-X)$, $G_t-X$ would not be
			compatible with $\mathcal{P}$ if $\mathcal{P}(i)=Q_w^*$ or $\mathcal{P}(j)=Q_w^*$. 
			Since $V_i(G_{t}-X)$ and $V_j(G_{t}-X)$ are of size at least two, $i$ and $j$ are
			assigned $Q_2$ or $Q_f$ in $\P$. 
			Both cases are covered by the transitions. 
			We deduce that there exists $F$ such that $F$ expresses the connectedness in
			$H(G_t-X,\mathcal{P})$, $(t,F,\mathcal{P})$ is reachable, and that $d[t,F,\mathcal{P}]
			\leq d[t',F',\mathcal{P}'] \leq |X|$.

			\item If $\mathcal{P}(i)=Q_\varnothing$ or $\mathcal{P}(j)=Q_\varnothing$, then
			$G_t-X = G_{t'}-X$ from compatibility and join definition. So we can apply induction
			hypothesis with $\mathcal{P}$ and follow the transition corresponding to $Q_\varnothing$
			to conclude.

			\item If $\P(i),\P(j) \in \{Q_1,Q_1^*\}$, $\mathcal{P}$ is compatible
			with $G_{t'}-X$ and $H(G_{t'}-X,\P) \subseteq H(G_t-X,\P)$ so $H(G_t-X,\P)$
			contains no $S$-cycle and no $S$-path between the representatives of $i$ and $j$. The
			induction hypothesis applied with $\mathcal{P}$ provides $F'$ such that $F'$ expresses
			the connectedness in $H(G_{t'}-X,\P')$, $(t',F',\mathcal{P})$ is reachable and
			$d[t',F',\mathcal{P}]\leq |X|$.
			Since $H(G_{t'}-X,\mathcal{P})$ contains no $S$-path between representatives of $i$ and $j$,
			and $F'$ expresses the connectedness in $H(G_{t'}-X,\P)$, we can follow the transition
			to get $F$ such that $F$ expresses the connectedness in $H(G_t-X,\P)$,
			$(t,F,\mathcal{P})$ is reachable and $d[t,F,\mathcal{P}] \leq d[t',F',\mathcal{P}'] \leq
			|X|$.
			
			\item If $\{\P(i),\P(j)\}=\{Q_1,Q_f\}$, w.l.o.g. let us assume that
			$\mathcal{P}(i)=Q_1$ and $\mathcal{P}(j)=Q_f$. We define $\mathcal{P}'$ and
			$\mathcal{P}''$ obtained from $\mathcal{P}$ with $\mathcal{P}'(j)=Q_w$ and
			$\mathcal{P}''(j)=Q_2$.

			If label $j$ contains no $S$-vertex in $G_{t'}-X$ then it is compatible with
			$\mathcal{P}''$, and $H(G_{t'}-X,\mathcal{P}'')$ is $S$-cycle-free because
			$H(G_t-X,\mathcal{P})$ is $S$-cycle-free and we can transform paths of
			$H(G_{t'}-X,\mathcal{P}')$ as described in the key observation. By applying the
			induction hypothesis with $\mathcal{P}''$, we get $F''$ such that $F''$ expresses the
			connectedness in $H(G_{t'}-X,\P'')$, $(t',F'',\mathcal{P}'')$ is reachable and
			$d[t',F'',\mathcal{P}''] \leq |X|$.  We can then follow the transition that accepts
			types $Q_1$ and $Q_2$ and conclude.
			
			Otherwise, $G_{t'}-X$ is compatible with $\mathcal{P}'$ and $H(G_{t'}-X,\mathcal{P}')$
			is $S$-cycle-free because $H(G_t-X,\mathcal{P})$ is $S$-cycle-free and we can transform
			paths of $H(G_{t'}-X,\mathcal{P}')$ as described in the key observation. By applying the
			induction hypothesis with $\mathcal{P}'$, we get $F'$ such that $F'$ expresses the
			connectedness in $H(G_{t'}-X,\P')$, $(t',F',\mathcal{P}')$ is reachable and
			$d[t',F',\mathcal{P}'] \leq |X|$.
			We can then follow the transition that accepts types $Q_1$ and $Q_w$ and conclude.
			
			\item If $\{\mathcal{P}(i),\mathcal{P}(j)\} = \{Q_1^*,Q_f\}$, w.l.o.g.
			let us assume that $\mathcal{P}(i)=Q_1^*$ and $\mathcal{P}(j)=Q_f$. Consider
			$\mathcal{P}'$ obtained from $\mathcal{P}$ by setting $\mathcal{P}'(j)=Q_w^*$. Since
			$H(G_t-X,\mathcal{P})$ is $S$-cycle-free, we get that $\mathcal{P}'$ is compatible with
			$G_t-X$ and that $H(G_{t'}-X,\mathcal{P}')$ is $S$-cycle-free. By induction hypothesis,
			there exists $F'$ such that $F'$ expresses the connectedness in $H(G_{t'}-X,\P')$,
			$(t,F',\mathcal{P}')$ is reachable and $d[t,F',\mathcal{P}'] \leq |X|$. We can follow
			the transition for types $Q_1^*$ and $Q_w^*$ and conclude.

			\item If one of $i$ and $j$ is in state $Q_1$ and the other is in state $Q_w^*$ or
			$Q_w$, the newly added path with endpoints in the label in state $Q_w^*$ makes it
			impossible for $\mathcal{P}$ to be compatible with $G_t-X$.

			\item If $\{\P(i),\P(j)\}=\{Q_1,Q_2\}$, we apply the induction hypothesis with
			$\mathcal{P}$ and follow the transition for types $Q_1$ and $Q_2$ to conclude.
			
			\item If one of $i$ and $j$ is in state $Q_1^*$ and the other is in state $Q_2$, $Q_w$
			or $Q_w^*$, then the join adds an $S$-path in $G_t-X$ with endpoints in the label in state 
			$Q_2$, $Q_w$ or $Q_w^*$ which makes it impossible for $\mathcal{P}$ to be compatible
			with $G_t-X$.
		\end{itemize} 
		
		\item Let $t$ be a renaming label node with $G_t= \rho_{i \rightarrow j}(G_{t'})$, let $X \subseteq
		V(G_t)$ and $\mathcal{P} \in \mathcal{Q}^k$ such that $\mathcal{P}$ is compatible with
		$G_t-X$ and $H(G_t-X,\mathcal{P})$ is $S$-cycle-free.
		Note that $\mathcal{P}(i)=Q_\varnothing$ by compatibility.
		In the following, by comparing paths that appear when constructing $F$ from $F'$ with paths
		that are in $H(G_t-X,\P)$ but not $H(G_{t'}-X,\P')$, we deduce that $F$ expresses the
		connectedness in $H(G_t-X,\P)$, this is already done in detail in the previous lemma.
		We consider the following cases based on $\mathcal{P}(j)$:
		\begin{itemize}	
			\item If $\mathcal{P}(j)=Q_f$, then consider $\mathcal{P}'$
			where $\mathcal{P}'(i)$ and $\mathcal{P}'(j)$ are determined by the type of the vertex
			with the corresponding label if it is unique in $G_{t'}-X$ and are equal to $Q_f$
			otherwise. Thus, $\mathcal{P}'$ is compatible with $G_{t'}-X$ and $H(G_{t'}-X,\mathcal{P}')$
			is $S$-cycle-free because only labels change with $H(G_t-X,\mathcal{P})$. By induction
			hypothesis, we have $F'$ such that $F'$ expresses the connectedness in
			$H(G_{t'}-X,\P')$, $(t',F',\mathcal{P}')$ is reachable and $d[t',F',\mathcal{P}'] \leq
			|X|$. We can thus use the corresponding transition and conclude.

			\item If $\P(j)=Q_1$ or $\P(j)=Q_1^*$, then only one of $i$ and $j$ contains a vertex in
			$G_{t'}-X$. We obtain $\mathcal{P}'$ from $\mathcal{P}$ by setting $\mathcal{P}'(i)$ and
			$\mathcal{P}'(j)$ according to which label contains the vertex (note that one of them
			will be $Q_\varnothing$). 
			Again, $\mathcal{P}'$ is compatible with $G_{t'}-X$ and $H(G_{t'}-X,\mathcal{P}')$ 
			is $S$-cycle-free because only labels change with $H(G_t-X,\mathcal{P})$. 
			By induction hypothesis, we have $F'$ such that $F'$ expresses the connectedness in
			$H(G_{t'}-X,\P')$, $(t',F',\mathcal{P}')$ is reachable and $d[t',F',\mathcal{P}'] \leq
			|X|$. We can follow the transition associated to type $Q_\varnothing$ and conclude.

			\item If $\P(j)=Q_2$, then the vertices labeled $i$ and $j$ in $G_{t'}-X$
			cannot contain $S$-vertices by compatibility. We obtain $\mathcal{P}'$ from
			$\mathcal{P}$ by setting $\mathcal{P}'(i)$ and $\mathcal{P}'(j)$ to 
			$Q_\varnothing$, $Q_1$ or $Q_2$ based on the number of vertices with the 
			corresponding label in $G_{t'}-X$. Again, $\mathcal{P}'$ is compatible with $G_{t'}-X$
			and $H(G_{t'}-X,\mathcal{P}')$ is $S$-cycle-free. There cannot be an $S$-path between
			representatives of $i$ and $j$ in $H(G_{t'}-X,\mathcal{P}')$ because
			$H(G_t-X,\mathcal{P})$ contains no $S$-cycle but has them connected by an additional
			path. By induction hypothesis, we get $F'$ such that $F'$ expresses the connectedness in
			$H(G_{t'}-X,\P')$, $(t',F',\mathcal{P}')$ is reachable and $d[t',F',\mathcal{P}']\leq
			|X|$. We can then follow a transition and conclude. 

			\item If $\P(j)=Q_w$, then consider $\mathcal{P}'$ obtained from $\mathcal{P}$
			by setting $\mathcal{P}'(i)$ and $\mathcal{P}'(j)$ to $Q_\varnothing$, $Q_1$, $Q_1^*$, 
			$Q_2$ or $Q_w$ based on the vertices with corresponding label in $G_{t'}-X$. $\mathcal{P}'$ 
			is compatible with $G_{t'}-X$, $H(G_{t'}-X,\mathcal{P}')$ is $S$-cycle-free because it
			is less connected than $H(G_t-X,\mathcal{P})$ and $H(G_{t'}-X,\mathcal{P}')$ cannot have
			an $S$-path between representatives of $i$ and $j$ because $H(G_t-X,\mathcal{P})$
			is $S$-cycle-free. By induction hypothesis, we get $F'$ such that $F'$ expresses the
			connectedness in $H(G_{t'}-X,\P')$, $(t',F',\mathcal{P}')$ and $d[t',F',\mathcal{P}']
			\leq |X|$. We can then follow a transition and conclude. 

			\item Finally, if $\mathcal{P}(j)=Q_w^*$, then consider $\mathcal{P}'$
			obtained from $\mathcal{P}$ by setting $\mathcal{P}'(i)$ and $\mathcal{P}'(j)$ to
			$Q_\varnothing$, $Q_1$, $Q_1^*$ or $Q_w^*$ based on the vertices with corresponding
			label in $G_{t'}-X$. $\mathcal{P}'$ is compatible with $G_{t'}-X$ because a path
			between its vertices labeled $i$ or $j$ would contradict $H(G_t-X,\mathcal{P})$ being
			$S$-cycle-free, $H(G_{t'}-X,\mathcal{P}')$ is $S$-cycle-free because it is less
			connected than $H(G_t-X,\mathcal{P})$ and there is no path between the representatives
			of $i$ and $j$ in $H(G_{t'}-X,\mathcal{P}')$ because it would also contradict
			$H(G_t-X,\mathcal{P})$ being $S$-cycle-free. By induction hypothesis, we get $F'$ such
			that $F'$ expresses the connectedness in $H(G_{t'}-X,\P')$, $(t',F',\mathcal{P}')$ is
			reachable and $d[t',F',\mathcal{P}'] \leq |X|$. We can then follow a transition and
			conclude.
		\end{itemize}

		\item Let $t$ be a disjoint union node with $G_t = G_{t_1} \oplus G_{t_2}$, let $X \subseteq
		V(G_t)$ and $\mathcal{P} \in \mathcal{Q}^k$ such that $\mathcal{P}$ is compatible with
		$G_t-X$ and $H(G_t-X,\mathcal{P})$ is $S$-cycle-free. 
		We define $X_1 = X \cap V(G_{t_1})$ and $X_2 = X \cap V(G_{t_2})$, and construct
		$\mathcal{P}_1$ and $\mathcal{P}_2$ by taking the value of $\mathcal{P}$ for their
		respective used labels, and taking value $Q_\varnothing$ otherwise.
		We immediately get that, for $i \in \{1,2\}$, $\mathcal{P}_i$ is compatible with
		$G_{t_i}-X_i$ and $H(G_{t_i}-X_i,\mathcal{P}_i)$ because used labels are disjoint.
		By induction hypothesis, we get $F_1$ and $F_2$ such that $F_1$ and $F_2$ express the
		connectedness in $H(G_{t_1}-X,\P_1)$ and $H(G_{t_2}-X,\P_2)$ respectively,
		$(t_1,F_1,\mathcal{P}_1)$ and $(t_2,F_2,\mathcal{P}_2)$ are reachable,
		$d[t_1,F_1,\mathcal{P}_1] \leq |X_1|$ and $d[t_2,F_2,\mathcal{P}_2]\leq |X_2|$.
		We can follow the transition and conclude.\qedhere
	\end{enumerate}
\end{proof}

\begin{lemma}\label{lem:val}
	The minimum value of a reachable state of the root of the expression is equal to the minimum size of a
	subset feedback vertex set.
\end{lemma}
\begin{proof}
	For $X^*$ a subset feedback vertex set of minimum size, we define $\mathcal{P}^*$ by
	$\mathcal{P}^*(i)=Q_\varnothing$ is $G_r-X^*$ has no vertex labeled $i$, $\mathcal{P}^*(i)=Q_f$ if
	there are at least 2 vertices labeled $i$ in $G_r-X^*$ and $\mathcal{P}^*(i)=Q_1$ or
	$\mathcal{P}^*(i)=Q_1^*$ based on the type of vertex labeled $a$ if there is exactly one.
	$\mathcal{P}^*$ is compatible with $G_r-X^*$ and $H(G_r-X,\mathcal{P})=G_r-X$ contains no $S$-cycle.
	By Lemma \ref{lem:cplt}, there is a state $(r,F^*,\mathcal{P})$ such that $d[r,F^*,\mathcal{P}]
	\leq |X^*|$. By Lemma \ref{lem:adm}, every reachable state has a value that corresponds to a subset
	feedback vertex set so by minimality of $|X^*|$, we have $d[r,F^*,\mathcal{P}]=|X^*|$ and
	$d[r,F^*,\mathcal{P}^*] \leq d[r,F,\mathcal{P}]$ for any other reachable state
	$(r,F,\mathcal{P})$.
\end{proof}

\begin{theorem}\label{thm:complexity}
	Given a $k$-expression describing graph $G$, \SFVS{} can be solved in time $2^{\mathcal{O}(k
	\log k)} \cdot n$, where $n$ is the size of the given $k$-expression.
\end{theorem}

\begin{proof}
	From Lemma \ref{lem:val} we know that the algorithm computes the solution to SFVS on $G$,
	For each $t \in T$, the number of states $(t,F,\mathcal{P})$ is $2^{\mathcal{O}(k \log k)}$
	since the forests have $\mathcal{O}(k)$ vertices by claim \ref{claim:sfvs-forest-size}.
	Note that enumerating pairs of states for disjoint union nodes takes time
	$\left(2^{\mathcal{O}(k \log k)}\right)^2=2^{\mathcal{O}(k \log k)}$. Summing for all
	nodes of $T$ (at most $n$), we conclude that the values for all states can be computed in time
	$2^{\mathcal{O}(k \log k)} \cdot n$.
\end{proof}

\section{Odd Cycle Transversal in graphs of bounded cliquewidth}\label{sec:oct}
\def\OCT{\textsc{Odd Cycle Transversal}}
\def\s{\mathbf{s}}
\def\u{\mathbf{u}}
\def\v{\mathbf{v}}
\def\x{\mathbf{x}}
\def\y{\mathbf{y}}
\def\cw{\mathbf{cw}}
\def\lcw{\mathbf{lcw}}

\subsection{The algorithm}

We describe a dynamic programming algorithm for \OCT{} parameterized by
clique-width.
In the following, we consider a $k$-labeled graph $G$, and assume that it comes with a
$k$-expression of size $n$ and its associated tree $T$ with root $r\in V(T)$.

A {\em state} in our dynamic programming algorithm is a tuple $(t,\x)$ where $t$ is a node of $T$, and $\x\in \{0,1,2,3\}^k$ is a vector.
It is called \emph{admissible} if there exists a triple $(X,A,B)$ with $X$ an OCT of $G_t$ and $(A,B)$ an ordered bipartition of $G_t-X$, called \emph{compatible} with $(t,\x)$, such that
\begin{itemize}
	\item $V_i(G_t)\subseteq X$ for every $i\in [k]$ with $\x[i]=0$;
	\item $V_i(G_t-X)\subseteq A$ for every $i\in [k]$ with $\x[i]=1$; and
	\item $V_i(G_t-X)\subseteq B$ for every $i\in [k]$ with $\x[i]=2$.
\end{itemize}
We call \emph{category} of label $i$ with respect to vector $\x$ the integer $\x[i]$.
Labels from category $3$ may contain any vertex of $G_t$.
Hence $X$ is not characterized by $\x$. 
In particular, two OCTs for a same graph $G_t$ may witness the admissibility of a single state $(t,\x)$: in $d[t,\x]$ will be stored the size of one such OCT.
We will later prove that the minimum size of an OCT for $G$ is found by iterating through the different values of states $(r,\x)$.

For convenience in the following, let us define the poset on ground set $\{0,1,2,3\}$ with
comparabilities ${0<1}$, ${0<2}$, ${1<3}$, and ${2<3}$, and consider the \emph{supremum} and
\emph{infimum} operations $\vee$ and $\wedge$ on such a poset: $a\vee b$ denotes the unique
minimal element $c$ such that both $a\leq c$ and $b\leq c$, $a \wedge b$ denotes the unique element
$c'$ such that both $a \geq c'$ and $b \geq c'$.
We extend the order relation to pairs of vectors $\x,\y \in \{0,1,2,3\}$ by defining $\x\leq \y$ if $\x[i]\leq \x[j]$ for all $i,j\in [k]$.
Similarly, the supremum of pairs of vectors
$\x,\y\in \{0,1,2,3\}^k$ is obtained by performing the supremum on each of their coordinates.

We describe a bottom-up computation of the values of states, handling each type of node in $T$. 
First, we initialize $d[t,\x]=+\infty$ for every node $t\in V(T)$ and vector $\x\in \{0,1,2,3\}^k$. 
We then proceed to the computation of some states, called \emph{reachable} in the remaining of the section, going trough the following transitions:

\begin{itemize}
	\item \textbf{Leaf node.} If $t$ is a leaf node with $G_t=i(v)$ for some vertex $v$ and label
	$i\in[k]$, we set $d[t,\{0\}^k]=c(v)$, $d[t,\x]=0$, and $d[t,\y]=0$, where $\x[i]=1$, $\y[i]=2$,
	and $\x[j]=\y[j]=0$ for $i\neq j\in [k]$.  These states correspond to either placing $v$ into an
	OCT we are constructing, or placing it into one of the two sides of a bipartition we are
	setting.

	\item \textbf{Disjoint union node.} Let $t$ be a disjoint union node with children $t_1,t_2$ such that ${G_t=G_{t_1} \oplus\,G_{t_2}}$. 
	For each $\x \in \{0,1,2,3\}^k$, we put
	$$d[t,\x] = \min \limits_{\x_1 \vee \x_2 \leq \x} d[t_1,\x_1]+d[t_2,\x_2].$$

	\item \textbf{Join node.} Let $t$ be a join node with child $t'$ such that $G_t=\eta_{i\times j}(G_{t'})$.
	We do not modify states but only let through those whose join can only preserve the bipartition induced by categories $1$ and $2$ in $G_{t'}$.
	Such states are those $(t',\x)$ which satisfy $\x[i] \wedge \x[j]=0$, for which we put
	$$d[t,\x] = d[t',\x].$$
	
	\item \textbf{Renaming label node.} Let $t$ be a renaming label node with child $t'$ such that $G_t=\rho_{i\rightarrow j}(G_{t'})$.
	For each reachable state $(t',\x')$ we create a vector $\x$ by setting $\x[i]=0$, $\x[j]=\x'[i] \vee \x'[j]$, and $\x[\ell]=\x'[\ell]$ for $\ell\not\in \{i,j\}$.
	We then put
	$$d[t,\x] \leftarrow d[t',\x'].$$
\end{itemize}

We show in the following lemma that reachable states are admissible and contain the value of a compatible OCT.

\begin{lemma}\label{lem:DP-exeq}
	For every reachable state $(t,\x)$ there exists a compatible triple $(X,A,B)$ such that $c(X)=d[t,\x]$.
\end{lemma}

\begin{proof}
	We proceed by induction on $T$.

	Let $t$ be a leaf node with $G_t=i(v)$ for some vertex $v$ and label $i\in[k]$, then three different states are initialized depending on whether $v$ is placed in the OCT, or in one of the two sides of the bipartition.
	Each of these states satisfies the conclusion.

	Let $t$ be a disjoint union node with children $t_1,t_2$ such that ${G_t=G_{t_1} \oplus\,G_{t_2}}$, and $(t,\x)$ be a reachable state.
	By the described transitions, there exist reachable states $(t_1,\x_1)$ and $(t_2,\x_2)$ such
	that $\x \geq \x_1 \vee \x_2$ and $d[t,\x]= d[t_1,\x_1]+d[t_2,\x_2].$
	By induction hypothesis, for $i\in [2]$, there exist triples $(X_i,A_i,B_i)$ compatible with
	$(t_i,\x_i)$ such that $c(X_i)=d[t_i,\x_i]$.
	As ${G_t=G_{t_1} \oplus\,G_{t_2}}$, $X=X_1\cup X_2$ is an OCT of $G_t$ and $(A_1\cup A_2,B_1\cup B_2)$ defines a bipartition of $G_t-X$.
	Furthermore as $\x \geq \x_1 \vee \x_2$, it is straightforward to check that $(X,A_1\cup A_2,B_1\cup B_2)$ are compatible with $(t,\x)$.
	We conclude noting that $c(X)=c(X_1)+c(X_2)=d[t_1,\x_1]+d[t_2,\x_2]=d[t,\x]$.

	Let $t$ be a join node with child $t'$ such that $G_t=\eta_{i\times j}(G_{t'})$, and $(t,\x)$ be a reachable state.
	By the described transitions $(t',\x)$ is also a reachable state satisfying one of the conditions $\x[i]=0$, or $\x[j]=0$, or ($\x[i]=1$ and $\x[j]=2$), or ($\x[i]=2$ and $\x[j]=1$).
	Furthermore, $d[t,\x]=d[t',\x]$.
	By induction hypothesis, there exists a triple $(X,A,B)$ compatible with $(t',\x)$ such that
	$c(X)=d[t',\x]$, and hence such that $c(X)=d[t,\x]$.
	According to the transition, as $(X,A,B)$ and $(t',\x)$ are compatible, every edge added to $G_{t'}$ when performing the join operation either has an endpoint in $X$, or one in $A$ and the other in $B$. 
	We deduce that $X$ is also an OCT of $G_t$, and that $(A,B)$ also defines a bipartition of $G_t$.
	Hence, $(X,A,B)$ is compatible with $(t,\x)$.

	Let $t$ be a renaming label node with child $t'$ such that $G_t=\rho_{i\rightarrow j}(G_{t'})$, and let $(t,\x)$ be a reachable state. 
	By the described transitions there exists $\x'$ such that $\x[i]=0$, $\x[j]=\x'[i] \vee \x'[j]$, and $\x[\ell]=\x'[\ell]$ for $\ell\not\in \{i,j\}$, and such that $d[t,\x]= d[t',\x']$.
	By induction hypothesis, there exists a triple $(X,A,B)$ compatible with $(t',\x')$ such that
	$c(X)=d[t',\x']$, and hence such that $c(X)=d[t,\x]$.
	As $G_t$ and $G_{t'}$ are isomorphic, $X$ is also an OCT of $G_{t}$, and $(A,B)$ also defines a bipartition of $G_{t}-X$.
	As $\x[i]=0$ and $\x[j]=\x'[i] \vee \x'[j]$ are the only modifications made to $\x'$, it is easy to check that $(X,A,B)$ is compatible with $(t,\x)$.
\end{proof}

\begin{lemma}\label{lem:DP-geq}
	For every node $t\in V(T)$, for every OCT $X$ of $G_t$ and bipartition $(A,B)$ of $G_t-X$, there
	exists $\x\in \{0,1,2,3\}^k$ such that $(t,\x)$ is reachable, compatible with $(X,A,B)$, and
	satisfies $c(X)\geq d[t,\x]$.
\end{lemma}

\begin{proof}
	We proceed by induction on $T$.

	Let $t$ be a leaf node with $G_t=i(v)$ for some vertex $v$ and label $i\in[k]$, then three different cases arise depending on whether $v$ is placed in the OCT, or in one of the two sides of the bipartition.
	For each of these cases, a compatible state is constructed and a value satisfying $c(X)\geq d[t,\x]$ is initialized.

	Let $t$ be a disjoint union node with children $t_1,t_2$ such that ${G_t=G_{t_1} \oplus\,G_{t_2}}$, $X$ be an OCT of $G_t$, and $(A,B)$ be a bipartition of $G_t-X$.
	As ${G_t=G_{t_1} \oplus\,G_{t_2}}$, $X_1=X\cap V(G_{t_1})$ and $X_2=X\cap V(G_{t_2})$ are OCTs
	of $G_{t_1}$ and $G_{t_2}$, and the ordered pairs
	$(A_1=A\cap V(G_{t_1}),B_1=B\cap V(G_{t_1}))$ and $(A_2=A \cap V(G_{t_2}),B_2 =
	B \cap V(G_{t_2}))$ define bipartitions of $G_{t_1}-X_1$ and $G_{t_2}-X_2$, respectively.
	By induction hypothesis, there exist reachable states $(t_1,\x_1)$ and $(t_2,\x_2)$ with
	$d[t_1,\x_1] \leq c(X_1)$ and $d[t_2,\x_2] \leq c(X_2)$ that are compatible with $(X_1,A_1,B_1)$ and $(X_2,A_2,B_2)$, respectively.
	By the described transitions, the state $(t,\x)$ defined by $\x:= \x_1 \vee \x_2$ is
	computed, and it is compatible with $(X,A,B)$. 
	Since $c(X)= c(X_1)+c(X_2)$ we deduce $d[t,\x] \leq d[t_1,\x_1]+d[t_2,\x_2]\leq c(X)$ as desired.

	Let $t$ be a join node with child $t'$ such that $G_t=\eta_{i\times j}(G_{t'})$, $X$ be an
	OCT of $G_t$, and $(A,B)$ be a bipartition of $G_t-X$. 
	Since $E(G_{t'}-X) \subseteq E(G_t-X)$, $X$ is also an OCT of $G_{t'}-X$, and $(A,B)$ is also
	a bipartition of $G_{t'}-X$.  
	By induction hypothesis, there exists a reachable state $(t',\x)$ that is compatible with
	$(X,A,B)$, and which satisfies $d[t',\x] \leq c(X)$.
	Since $(A,B)$ is a bipartition of $G_t-X$, the edges that were added to $G_{t'}$ by the join operation either had an endpoint in $X$, or one in $A$ and the other in $B$. 
	Hence, either one of the labels $i$ and $j$ is empty (category 0), or $i$ and $j$ are on distinct sides of the partition. 
	We conclude to a transition from $(t',\x)$ to $(t,\x)$, and hence $d[t,\x] \leq d[t',\x]\leq
	c(X)$.

	Let $t$ be a renaming label node with child $t'$ such that $G_t=\rho_{i\rightarrow j}(G_{t'})$,
	$X$ be an OCT of $G_t$ and $(A,B)$ be a bipartition of $G_t-X$.
	Since $G_t$ and $G_{t'}$ are isomorphic, $X$ is also an OCT of $G_{t'}$, and $(A,B)$ is also a bipartition of $G_{t'}-X$. 
	By induction hypothesis, there exists a reachable state $(t',\x')$ that is compatible with
	$(X,A,B)$, and which satisfies $d[t',\x'] \leq c(X)$. 
	The transition from this state gives us a reachable state $(t,\x)$ such that $d[t,\x] \leq
	d[t',\x'] \leq c(X)$.
	The obtained state $(t,\x)$ is still compatible with $(A,B)$, as its modification exactly translates the change of label of the concerned vertices.
\end{proof}

We deduce the correctness of our algorithm.

\begin{theorem}\label{thm:DP-correctness}
	The minimum size of an OCT for $G$ is equal to the minimum, among $\x\in \{0,1,2,3\}^k$, of $d[r,\x]$.
\end{theorem}

\begin{proof}
	Let $X^*$ be an OCT of $G$ of minimum size, and $(A,B)$ be a bipartition of $G-X^*$.
	By Lemma~\ref{lem:DP-geq}, there exists a reachable state $(r,\x)$ compatible with $(X^*,A,B)$
	such that $c(X^*)\geq d[r,\x]$.
	By Lemma~\ref{lem:DP-exeq}, for every other state $(r,\y)$, either $(r,\y)$ is not reachable and
	$d[r,\y]=+\infty$, or it is reachable and there exists an OCT $Y$ of $G$ such that $c(Y)=d[r,\y]$.
	We conclude by noting that $c(X^*)\leq c(Y)$ for every such set $Y$.
\end{proof}

We conclude the section with the time analysis.

\begin{theorem}
	There are at most $\mathcal{O}(4^k\cdot n)$ possible states in total, and computing
	all reachable states takes $\mathcal{O}(4^k\cdot k \cdot n)$ time.
\end{theorem}

\begin{proof}
	There are $4^k$ states per node of the $k$-expression, hence there are at most
	$\mathcal{O}(4^k\cdot n\cdot k^2)$ possible states in total in the algorithm.

	The computation of a leaf node takes $\Oh(1)$ time, and the computations of a join or renaming
	label node is done within $\mathcal{O}(4^k)$ time.
	Let us show that the computation of a disjoint union node $t$ takes $\mathcal{O}(4^k \cdot k)$ time. 
	We first compute
	$\delta_i(\x) = \min \{ d[t_i,\x'] : \x' \leq \x \}$ for $i \in [2]$ and every $\x \in \{0,1,2,3\}^k$.
	This is done in $\mathcal{O}(4^k \cdot k)$ time starting from $\delta_i(\{0\}^k)=d[t_i,\{0\}^k]$, and computing the other values from bottom to top, observing that, by monotonicity of $\delta_i$, $\delta_i(\x)=\min \{d[t_i,\x], \delta_i(\x') \colon \x'\in \pred(\x)\}$, where $\pred(\x)$ denote the immediate predecessors of $\x$ in $\{0,1,2,3\}^k$ with respect to $\leq$.
	Now, note that by definition of the supremum, we have the following identity:
	$$\min_{\x_1 \vee \x_2 \leq x} d[t_1,\x_1] + d[t_2,\x_2] = \min_{\x_1 \leq \x}
	d[t_1,\x_1] + \min_{\x_2 \leq \x} d[t_2,\x_2].$$	
	We deduce that $d[t,\x]=\delta_1(\x) + \delta_2(\x)$.

	Summing up, our algorithm runs in time $\mathcal{O}(4^k\cdot k \cdot n)$.
\end{proof}

\subsection{The lower bound}

Our construction follows the line of the ones for \textsc{Independent Set} and \textsc{Odd Cycle Transversal} in \cite{lokshtanov2011known}.
Let $\varphi$ be an instance of SAT on $n$ variable and $m$ clauses.
Note that we may assume that $n$ is even, as otherwise we can add a single extra dummy variable to the instance with no impact on the satisfiability of $\varphi$.
We denote by $\rho$ the sum of the sizes of the clauses in $\varphi$, i.e., the total number of occurrences of literals in $\varphi$.

Let us first describe a gadget introduced in \cite[Section 7]{lokshtanov2011known}.
For two vertices $u,v$ we call \emph{arrow from $u$ to $v$} and denote by $A(u,v)$ the graph consisting of a path $ua_1a_2a_3v$ on five vertices, together with the additional four vertices $b_1, b_2, b_3,b_4$ and eight edges $ub_1, b_1a_1, a_1b_2, b_2a_2, a_2b_3, b_3a_3, a_3b_4, b_4v$. 
This graph is illustrated in Figure~\ref{fig:arrow}.
Note that such a graph has a unique smallest OCT of size $2$, which is $\{a_1,a_3\}$:
we call \emph{passive OCT} of the arrow such an OCT.
In $A(u,v)\setminus \{u\}$, $\{a_2, v\}$ is a smallest OCT:
we call it \emph{active OCT} of the arrow.

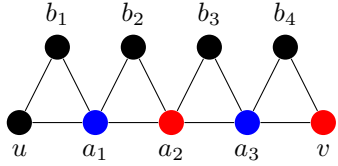
\begin{figure}[H]
	\centering
	\begin{tikzpicture}
		\foreach \i in {1,...,4} { \node[circle,fill,label=above:$b_\i$] (b\i) at (\i-0.5,1) {};}
		\foreach \i in {1,3} { \node[circle,fill,blue,label=below:$a_\i$] (a\i) at (\i,0) {};}
		\node[circle,fill,label=below:$u$] (a0) at (0,0) {};
		\node[circle,fill,red,label=below:$a_2$] (a2) at (2,0) {};
		\node[circle,fill,label=below:$v$,red] (a4) at (4,0) {};
		\foreach \i in {0,...,3} {
			\draw (a\i) -- (b\number\numexpr\i+1\relax); 
			\draw (a\i) -- (a\number\numexpr\i+1\relax);
			\draw (a\number\numexpr\i+1\relax) -- (b\number\numexpr\i+1\relax);
		}
	\end{tikzpicture}
	\caption{The arrow $A(u,v)$. Vertices in blue form the passive OCT of $A(u,v)$, and those in red form the active OCT of $A(u,v) \setminus \{u\}$.}
	\label{fig:arrow}
\end{figure}

For a clause $C=\{\ell_1,\dots,\ell_t\}$ with $t$ an odd integer greater than $2$, let us denote by $\widehat{C}$ the graph consisting of a simple cycle on vertex set $\{c_1,\dots,c_t\}$ with $c_i$ representing literal $\ell_i$, and with edges linking $c_i$'s in the natural way: $c_ic_{i+1}$ for every $i\in [t-1]$, and $c_tc_{1}$.
For a clause $C$ of even size $t\geq 2$, we proceed similarly but subdivide an arbitrary edge so that the obtained graph $\widehat{C}$ is an odd cycle.
For a clause $C$ of size one, we simply create a triangle $\widehat{C}$, one of its vertices representing the unique literal of the clause.

We construct a graph $G_1$ as follows.
We create $n$ paths $P_1,\dots,P_n$ each of length $2m$, and denote by $p_{i,j}$ the $j^{\text{th}}$ vertex on path $P_i$.
Each $P_i$ will be referred to as the path of variable $x_i$ in the following.
We add ``crossing'' edges as follows: for every odd $i\in [n-1]$ and even $j\in [2m-1]$, we add edges $p_{i,j}p_{i+1,j+1}$ and $p_{i+1,j}p_{i,j+1}$;
this is illustrated in Figure~\ref{fig:cross}.
Crossing edges will serve no other purpose than allowing to construct these paths with $n/2$ labels.
At this stage, paths of $G_1$ may be seen as rows representing variables of $\varphi$, and vertices
$\bigcup_{i\in [n]} \{p_{i,2j}, p_{i,2j-1}\}$ for $j\in [m]$ may be seen as columns grouped two by
two representing the two possible values variables $x_i$ may take in clause $C_j$.
For every clause $C_j=\{\ell_1,\dots,\ell_{|C_j|}\}$ in $\varphi$, we add cycle $\widehat{C}_j$ to $G_1$, and make the following connections.
If variable $x_i$ appears positively in $C_j$ on position $p$, i.e., $\ell_p=x_i$, then we add arrow $A(p_{i,2j-1}, c_p)$ to $G_1$.
If $x_i$ appears negatively in $C_j$ on position $p$, then we add arrow $A(p_{i,2j}, c_p)$ to $G_1$.
This concludes the construction of the graph~$G_1$.

We now consider $n+1$ copies of $G_1$ that we name $G_1,\dots,G_{n+1}$.
Let us denote by $P^k_1,\dots,P^k_n$ the $n$ paths of $G_k$, $p^k_{i,j}$ the $j^\text{th}$ vertex of $P^k_i$, and $C^k$ the cycle representing clause $C$ in $G_k$, for every $k\in [n+1]$.
We link consecutive graphs $G_k,G_{k+1},k\in [n]$ by adding edges $p^k_{i,2m}p^{k+1}_{i,1}$ for
every $i\in [n]$, and denote by $P^*_i$ the resulting path of length $2m\cdot (n+1)$ obtained by concatenating $P^1_i,\dots,P^{n+1}_i$ in such a way.
In addition, we add crossing edges $p^k_{i,2m}p^{k+1}_{i+1,1}$ and $p^k_{i+1,2m}p^{k+1}_{i,1}$ for every $k\in [n]$ and odd $i\in [n]$.
We complete our construction by creating a biclique of bipartition $\{A,B\}$ with $A$ and $B$ of size $\alpha+1=(n+1)(nm+2\rho)+1$, by connecting every vertex of $A$ to every vertex of $P^*_i$ for odd $i\in [n]$, and by connecting every vertex of $B$ to every vertex of $P^*_i$ for even $i\in [n]$.
We call $G$ the obtained graph.

\begin{figure}[H]
	\centering
	\begin{tikzpicture}
	\foreach \i in {0,...,5}
	{
		\foreach \j in {0,...,4}
		{
			\node[circle,fill] (\j\i) at (\j+3,\i) {};
		}
		\node (e\i) at (7.5,\i) {};
		\foreach \j in {0,...,3}
		{
			\draw (\j\i) -- (\number\numexpr \j +1\relax\i);
		}
		\draw (3\i) -- (e\i);
		\node (dots\i) at (8,\i) {$\dots$};
	}
	\foreach \i in {0,2,4}
	{
		\foreach \j in {1,3}
		{
			\draw (\j\i) -- (\number\numexpr \j + 1 \relax\number\numexpr \i + 1\relax);
			\draw (\j\number\numexpr\i+1\relax) -- (\number\numexpr \j+1\relax\i);
		}
	}
	\foreach \i in {0,...,5}
	{
		\foreach \j in {6,7}
		{
			\node[circle,fill] (\j\i) at (\j+3,\i) {};
		}
		\node (b\i) at (8.5,\i) {};
		\draw (b\i) -- (6\i) -- (7\i);	
	}
	\foreach \i in {0,2,4}
	{
		\node (c\i) at (8.5,\i+0.4) {};
		\draw (c\i) -- (6\i);
	}
	\foreach \i in {1,3,5}
	{
		\node (c\i) at (8.5,\i-0.4) {};
		\draw (c\i) -- (6\i);
	}
	\end{tikzpicture}
	\caption{The crossing edges on paths $P_1,\dots,P_n$ in $G_1$.}
	\label{fig:cross}
\end{figure}
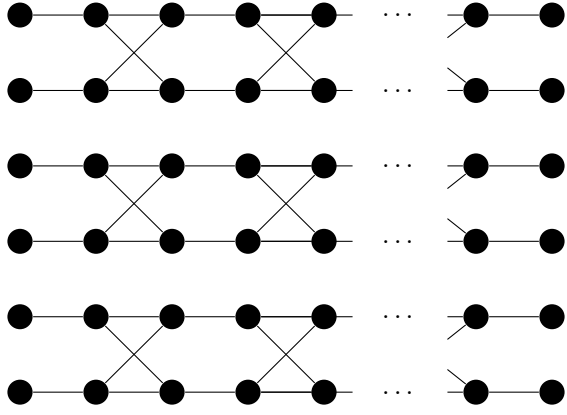

\begin{lemma}
	If $\varphi$ is satisfiable, then there exists an OCT of size $\alpha=(n+1)(nm + 2\rho)$ in $G$.
\end{lemma}

\begin{proof}
	Let $\tau$ be a satisfying truth assignment of $\varphi$. 
	We construct an OCT $T$ of $G$ as follows.
	If $\tau(x_i)$ is true, then we add $p^k_{i,j}$ to $T$ for every $k\in [n+1]$ and odd $j\in [2m]$.
	Otherwise we add $p^k_{i,j}$ to $T$ for every $k\in [n+1]$ and even $j\in [2m]$.
	For each of the $\rho\cdot(n+1)$ arrows $A(u,v)$ in $G$, we add its active OCT if $u$ is part of $T$, and its passive OCT otherwise, each of size two.
	The obtained set $T$ is of size $\alpha$.

	Since each clause $C$ is satisfied by $\tau$, at least one of its literals is satisfied.
	Thus for each cycle in $\widehat{C}^k, k\in [n+1]$ there exists an arrow $A(u,v)$ having an endpoint $u$ being part of $T$.
	Hence that arrow is active and we conclude that $v$ as well is part of $T$.
	We conclude that every (odd) cycle $\widehat{C}^1,\dots,\widehat{C}^{n+1}$ corresponding to $C$ is hit by $T$.
	Since $T$ is a vertex cover of the paths $P^*_i, i\in [n]$ every remaining edge in $G-T$ that is
	not part of a clause cycle nor an arrow is either a crossing edge, or an edge incident to a
	vertex of $A$ or $B$.
	Since these edges do not induce odd cycles, we conclude that $T$ is an OCT of $G$, as desired.
\end{proof}

\begin{lemma}
	If there exists an OCT of size $\alpha=(n+1)(nm + 2\rho)$ in $G$, then $\varphi$ is satisfiable.
\end{lemma}

\begin{proof}
	First note that an OCT of $G$ must either contain a vertex cover of each $P^*_i, i\in[n]$ or contain all of $A$ or $B$.
	Since $|A|=|B|>\alpha$, we can assume the existence of an OCT $T$ of size $\alpha$ satisfying the first condition.
	Hence $T$ contains at least half of the vertices of paths $P^*_i, i\in[n]$, that is, at least $nm\cdot (n+1)$ vertices in total.
	Since we need at least two additional vertices to intersect the odd cycles of every arrow (even if an endpoint of the arrow is part of $T$), and that there are $\rho$ of them, we conclude that $T$ contains precisely $nm\cdot (n+1)$ vertices from $P^*_1,\dots,P^*_n$, and $2\rho\cdot (n+1)$ vertices from the arrows.
	Now, observe that a vertex cover of size $k$ in a path of size $2k$ may contain at most once two consecutive vertices.
	Since $P^*_i$'s are made from $n+1$ repetitions of $n$ paths, there exists $G_j, j\in [n+1]$ for which none of the paths $P^j_i, i\in [n]$ contains consecutive vertices from $T$.
	We construct a truth assignment $\tau$ for $\varphi$ out of these paths as follows.
	For a variable $x_i, i\in [n]$ we set $\tau(x_i)$ to true if the vertices of $T$ coincide with odd vertices of $P^j_i$, and $\tau(x_i)$ to false otherwise.

	Since every clause cycle in $G_j$ has to be hit by $T$, and $T$ already contains $\alpha$ vertices from $P^*_1,\dots,P^*_n$ and the arrows, we conclude that each clause cycle in $G_j$ contains at least one incoming arrow $A(u,v)$ with $v\in T$.
	Hence $T$ contains at least three vertices from that arrow; see Figure~\ref{fig:arrow}. 
	Since $T$ intersects $A(u,v)-\{u\}$ on at most two vertices, the only way to hit every cycle of $A(u,v)$ that way is for $T$ to contain $\{u,a_2,v\}$: the arrow is active and $T$ contains $u$.
	By construction, the variable corresponding to vertex $u$ satisfies the clause in~$\varphi$.
	We conclude to a satisfying truth assignment of $\varphi$ as desired. 
\end{proof}

\begin{lemma}
	The constructed graph $G$ satisfies $\cw(G) \leq \lcw(G) \leq n/2 + 10$.
\end{lemma}

\begin{proof}
	We describe a construction of $G$ by a linear $K$-expression with $K=n/2 + 10$ labels. This
	proves $\lcw(G) \leq K$, and $\cw \leq \lcw$ follows from definition.

	During the construction, $n/2$ labels will be used to maintain the current path extremities, grouped two by two.
	These labels are indexed by $[n/2]$.
	Two labels will be dedicated to build clause cycles, denoted $a$ and $b$. 
	No more than five labels will be used to build arrows, and to prolongate current paths extremities.
	These labels are referred to as \emph{working labels} in the following.
	At last, two labels will be used for the biclique induced by $A$ and $B$, and one \emph{final} label indexed by $0$ will contain the rest of the graph: no further modifications will be made to vertices once they are put in that~label.
	We first construct the biclique with its two dedicated labels.
	Then, we will construct the rest of $G$ from $G_1$ to $G_{n+1}$, column after column (recall that columns consist of groups of two vertices from each $P^*_i, i\in [n]$), with each column constructed from top to bottom two rows at a time.
	At each step, four vertices are added to the $P^*_i$'s, and the gadgets are constructed along the way.

	Let us focus on the graph $G_1$ and assume that we are currently constructing the column
	described by indices $j,j+1$ for an odd $j\in [2m]$, and are willing to construct rows $i,i+1$ of that column for an odd $i\in [n]$.
	First, we create vertices $p_{i,j}$ and $p_{i+1,j}$ using one distinct working label for each, call them $y_1$ and $y_2$, respectively.
	If one of $p_{i,j}$ and $p_{i+1,j}$, name it $u$, participates to an arrow $A(u,c_k)$ linking $u$ to a clause cycle $\widehat{C}$ on vertex $c_k, k\in [|C|]$, we construct that arrow using other working labels.
	Clearly, three additional labels suffice to proceed with the construction of the arrow (not
	counting label 0).
	The vertices of the clause cycles will in fact be constructed by adding arrows one after the other.
	When creating arrow $A(u,c_k)$, $c_k$ is joined with label $a$, the vertex labeled $a$ is
	relabeled 0, and $c_k$ is relabeled $a$. We point out that the order of the vertices in the
	cycle does not matter.
	Note that to this extent, we only need three labels for $\widehat{C}$: two labels $a$ and $b$
	for the extremities of paths that constitute the cycle for now, and label $0$ for internal vertices
	in these path (they will not be joined anymore).
	We point out that during the construction of the column described by $j,j+1$, only the clause
	cycle $\widehat{C}$ is constructed.
	When finishing to build the column, the possible extra vertices that were added to $\widehat{C}$
	so that it is of odd size are added by joins between $a$ and a working label, emptying label
	$a$, and labeling them $a$.
	Once all vertices of $\widehat{C}$ have been added, the cycle is closed by joining labels $a$
	and $b$, which are then relabeled to $0$ and hence, labels $a$ and $b$ are free to be used for
	the next column.

	When constructing an arrow, we move every internal vertex in $A(u,v)\setminus \{u,v\}$ that has
	edges to all of its neighbours into label $0$.
	Then, we join $p_{i,j}$ and $p_{i+1,j}$ to the extremities of paths $P^*_i$,$P^*_{i+1}$ if they
	exist by joining label $\left\lceil i/2\right\rceil$ (containing the extremities of these paths only) to $y_1$ and $y_2$.
	This yields the crossing edges described in Figure~\ref{fig:cross}.
	We then rename label $\left\lceil i/2 \right\rceil$ to 0.
	Lastly, we join $p_{i,j}$ and $p_{i+1,j}$ accordingly to $A$ and $B$.

	The creation of vertices $p_{i,j+1}$ and $p_{i+1,j+1}$ follows the same principle regarding the gadgets.
	Only we join them to $p_{i,j}$ and $p_{i+1,j}$ using working labels $z_1$ and $z_2$ for
	$p_{i,j+1}$ and $p_{i+1,j+1}$, respectively, and joining $y_1$ to $z_1$ and $y_2$ to $z_2$. This
	is followed by renaming $y_1$ and $y_2$ to $0$.
	The vertices labeled $z_1$ and $z_2$ can each be renamed to $\left\lceil i/2\right\rceil$ as
	soon as joins to $y_1$ and $y_2$ respectively, and to neighbours in the possible arrow are done.

	Repeating this process to the next row, column, until $G_{n+1}$ is constructed yields our graph $G$ with the desired number of labels. 
\end{proof}

Observe that here we deduced a lower bound for the parameter linear clique-width from a lower bound
for the parameter pathwidth. Both parameters are respectively a linearized version of clique-width
and a linearized version of treewidth, the parameters of the algorithms for which we prove tight
complexity.

\appendix

\section{Appendix}

\begin{lemma}\label{lem:join-unique-correctness}
	Given a $k$-expression of a graph $G$ with associated tree $T$, if an edge of $j$ appears in two join operations, then in $T$ the nodes associated to the operations satisfy that one is
	an ancestor of the other. If we denote them $a,b \in V(T)$ with $a$ the ancestor, then $b$ can be
	replaced by its child.
\end{lemma}

\begin{proof}
	Suppose edge $uv$ appears in two join operations associated to nodes $t,t' \in V(T)$.
	
	Because vertices are introduced only once in a $k$-expression, the vertices $s\in V(T)$ satisfying $u\in V(G_s)$ are exactly the ancestors of the node $s'\in V(T)$ creating $u$. 
	Because $t$ and $t'$
	must contain $u$, this means that one of $t$ and $t'$ is an ancestor of the other. We denote $a$
	the ancestor and $b$ the descendant.
	Denote by $i$ the label of $u$ in $a$ and $i'$ its label in $b$. Denote by $j$ the label of $v$
	in $a$ and $j'$ its label in $b$.

	Note that because there is no operation that separates vertices that have the same label,
	$V_{i'}(G_b) \subseteq V_i(G_a)$ and $V_{j'}(G_b) \subseteq V_j(G_a)$. 
	As $uv$ appears in the join induced by nodes $a$ and $b$, it must be that the performed joins are $\eta_{i \times j}$ and $\eta_{i' \times j'}$, respectively.
	These joins introduce edges $A=\{ xy : x \in V_i(G_a),\ y \in V_j(G_a)\}$ and $B=\{ xy : x \in
	V_{i'}(G_b),\ y \in V_{j'}(G_b)\}$.
	From the previous inclusions, we deduce $B \subseteq A$. Hence, if $b$ is replaced by its child
	in $T$, which removes $b$ from the tree of the expression, the constructed graph is the
	unchanged.
\end{proof}

\begin{lemma}\label{lem:join-unique-algo}
	Given a $k$-expression of a graph $G$, and its associated tree $T$, there is an algorithm running in time $\Oh(k^2 \cdot |V(T)|)$ that produces a $k$-expression of $G$ with associated tree $T'$ in which every edge of $G$ appears once in a join node, and such that each join node introduces an edge (i.e., it does not involve an empty label).
\end{lemma}

\begin{proof}
	We will make two traversals of $T$: the first one will determine the nodes to be removed in $T$, while the second will compute the desired $k$-expression tree $T'$.
	Let $R$ be a table of Boolean variables indexed by vertices of $T$, and which will indicate vertices to be deleted from $T$; this variable is global and has its values initialized to $\bot$.
	The first traversal of $T$ is described by the function \texttt{FindRepetition}, which, given a node $t\in V(T)$, computes for each pair of labels $\{i,j\}$ a list $L(i,j)$ of join nodes of $T$, descendant of $t$ and taken maximal w.r.t.~the ancestor-descendant relation $\leq$, involving edges between vertices that are currently in labels $i$ and $j$.
	Recall that by Lemma~\ref{lem:join-unique-correctness}, these joins are in fact redundant: every edge involved in one such join is an edge involved in the join of the current node $t$.
	The function also computes, for each label $i$, a Boolean value $B[i]$ indicating whether label $i$ is nonempty.
	This will be used in order to detect and remove trivial join involving an empty label.
	
	We now describe the computation of \texttt{FindRepetition} according to each type of node in the $k$-expression.
	\begin{description}
		\item[Leaf node.] If $t=i(v)$ is a leaf node, we set:
		\begin{itemize}
			\item $L(a,b):=\varnothing$ for all $a,b\in [k]$;

			\item $B[i]=\top$; and

			\item $B[j]=\bot$ for all $j\neq i$.
		\end{itemize}			

		\item[Disjoint node.] If $t=t_1 \oplus t_2$, we call \texttt{FindRepetition} on $t_1$ and $t_2$, respectively returning $(L_1,B_1)$ and $(L_2,B_2)$.
		Then, we set:
		\begin{itemize}
			\item $L(i,j):= L_1(i,j) \sqcup L_2(i,j)$ for all $i,j\in [k]$, where $\sqcup$ denotes the list concatenation; and

			\item $B[i]:=B_1[i] \cup B_2[i]$ for all $i\in [k]$.
		\end{itemize}	
		
		\item[Join node.] If $t=\eta_{i \times j}(t')$, we call \texttt{FindRepetition} on $t'$, returning $(L',B')$. We then set $B:=B'$, and proceed as follows:
		\begin{itemize}
			\item for $\{a,b\} \neq \{i,j\}$, we set $L(a,b):=L'(a,b)$;

			\item for $\{i,j\}$, if $B[i]=\bot$ or $B[j]=\bot$, we set $R[t]:=\top$ and $L(i,j):=\varnothing$; and

			\item otherwise, we set $L(i,j):=(t)$, and for each $s \in L'(i,j)$, we set $R[s]:=\top$.
		\end{itemize}
		
		\item[Renaming label node.] If $t=\rho_{i \rightarrow j}(t')$, we call \texttt{FindRepetition} on $t'$ and gets $(L',B')$ as its return. Then:
		\begin{itemize}
			\item for $\{a,b\}$ such that $i,j \notin \{a,b\}$, we set $L(a,b):=L'(a,b)$;

			\item for $\{a,i\}$, we set $L(a,i):=\varnothing$.

			\item for $\{a,j\}$, $a \neq i$, we set $L(a,i):=L'(a,i) \sqcup L'(a,j)$;

			\item we set $B[i]:=\bot$, $B[j]:=B'[i] \cup B'[j]$; and

			\item for all $a \in [k] \setminus \{i,j\}$, we set $B[a]:=B'[a]$.
		\end{itemize}
	\end{description}

	The second traversal of the tree simply consists in replacing nodes $t$ such that $R[t]=\top$ by their child, as in Lemma~\ref{lem:join-unique-correctness}.

	For the complexity analysis, first note that using double linked for $L(i,j)$, attributions and concatenations may be performed in $\Oh(1)$ time.
	Parsing $L(i,j)$ in join nodes takes at most $\Oh(n)$ time in total, because each node is put in at most one list, and the list is discarded after it is parsed.
	The conclusion follow from the fact that there are $\Oh(k^2)$ pairs of labels, and $n$ nodes in $T$.

	In order to prove the correctness, we make the following observations.
	\begin{itemize}
		\item If $R[t]=\top$, $t$ is either a trivial join node (involving an empty label), or a join node that can be replaced by its child from $T$ by Lemma
		\ref{lem:join-unique-correctness}.
		\item At node $t$, $B[i]=\top$ if and only if there is a vertex labeled $i$.
		\item At node $t$, $L(i,j)$ contains exactly join nodes $s$ that added an edge between a vertex
		labeled $i$ for $t$ and a vertex labeled $j$ for $t$ and such that $R[s]=\bot$.
		\item At join node $t$ joining $i$ and $j$, for all join nodes $s$ under $t$ that added an edge
		between a vertex labeled $i$ for $t$ and a vertex labeled $j$ for $t$, $R[s]=\top$, and if
		$R[t]=\bot$, it adds at least one edge.
	\end{itemize}
	It is straightforward to prove these claims by induction on $t$. From them we can conclude.
\end{proof}

\bibliographystyle{alpha}

\bibliography{main}

\newcommand{\etalchar}[1]{$^{#1}$}
\begin{thebibliography}{BKN{\etalchar{+}}18}

\bibitem[BBBK20]{bergougnoux2020close}
Benjamin Bergougnoux, {\'E}douard Bonnet, Nick Brettell, and {O-joung} Kwon.
\newblock {Close Relatives of Feedback Vertex Set Without Single-Exponential
  Algorithms Parameterized by Treewidth}.
\newblock In Yixin Cao and Marcin Pilipczuk, editors, {\em 15th International
  Symposium on Parameterized and Exact Computation (IPEC 2020)}, volume 180 of
  {\em Leibniz International Proceedings in Informatics (LIPIcs)}, pages
  3:1--3:17, Dagstuhl, Germany, 2020. Schloss Dagstuhl--Leibniz-Zentrum f{\"u}r
  Informatik.

\bibitem[BCKN15]{DBLP:journals/iandc/BodlaenderCKN15}
Hans~L. Bodlaender, Marek Cygan, Stefan Kratsch, and Jesper Nederlof.
\newblock Deterministic single exponential time algorithms for connectivity
  problems parameterized by treewidth.
\newblock {\em Inf. Comput.}, 243:86--111, 2015.

\bibitem[BKN{\etalchar{+}}18]{DBLP:conf/wg/BonamyKNPSW18}
Marthe Bonamy, Lukasz Kowalik, Jesper Nederlof, Michal Pilipczuk, Arkadiusz
  Socala, and Marcin Wrochna.
\newblock On directed feedback vertex set parameterized by treewidth.
\newblock In Andreas Brandst{\"{a}}dt, Ekkehard K{\"{o}}hler, and Klaus Meer,
  editors, {\em Graph-Theoretic Concepts in Computer Science - 44th
  International Workshop, {WG} 2018, Cottbus, Germany, June 27-29, 2018,
  Proceedings}, volume 11159 of {\em Lecture Notes in Computer Science}, pages
  65--78. Springer, 2018.

\bibitem[BST19]{DBLP:journals/corr/abs-1907-04442}
Julien Baste, Ignasi Sau, and Dimitrios~M. Thilikos.
\newblock Hitting minors on bounded treewidth graphs. {IV.} an optimal
  algorithm.
\newblock {\em CoRR}, abs/1907.04442, 2019.

\bibitem[BST20a]{DBLP:journals/siamdm/BasteST20}
Julien Baste, Ignasi Sau, and Dimitrios~M. Thilikos.
\newblock Hitting minors on bounded treewidth graphs. {I.} general upper
  bounds.
\newblock {\em {SIAM} J. Discret. Math.}, 34(3):1623--1648, 2020.

\bibitem[BST20b]{DBLP:journals/tcs/BasteST20}
Julien Baste, Ignasi Sau, and Dimitrios~M. Thilikos.
\newblock Hitting minors on bounded treewidth graphs. {II.} single-exponential
  algorithms.
\newblock {\em Theor. Comput. Sci.}, 814:135--152, 2020.

\bibitem[BST20c]{DBLP:journals/jcss/BasteST20}
Julien Baste, Ignasi Sau, and Dimitrios~M. Thilikos.
\newblock Hitting minors on bounded treewidth graphs. {III.} lower bounds.
\newblock {\em J. Comput. Syst. Sci.}, 109:56--77, 2020.

\bibitem[BSTV13]{Bui-XuanSTV13}
Binh{-}Minh Bui{-}Xuan, Ondrej Such{\'{y}}, Jan~Arne Telle, and Martin
  Vatshelle.
\newblock Feedback vertex set on graphs of low clique-width.
\newblock {\em Eur. J. Comb.}, 34(3):666--679, 2013.

\bibitem[CER93]{CourcelleER93}
Bruno Courcelle, Joost Engelfriet, and Grzegorz Rozenberg.
\newblock Handle-rewriting hypergraph grammars.
\newblock {\em J. Comput. Syst. Sci.}, 46(2):218--270, 1993.

\bibitem[CFK{\etalchar{+}}15]{cygan2015parameterized}
Marek Cygan, Fedor~V Fomin, {\L}ukasz Kowalik, Daniel Lokshtanov, D{\'a}niel
  Marx, Marcin Pilipczuk, Micha{\l} Pilipczuk, and Saket Saurabh.
\newblock {\em Parameterized algorithms}, volume~5.
\newblock Springer, 2015.

\bibitem[CKN18]{DBLP:journals/jacm/CyganKN18}
Marek Cygan, Stefan Kratsch, and Jesper Nederlof.
\newblock Fast hamiltonicity checking via bases of perfect matchings.
\newblock {\em J. {ACM}}, 65(3):12:1--12:46, 2018.

\bibitem[CMPP17]{DBLP:journals/iandc/CyganMPP17}
Marek Cygan, D{\'{a}}niel Marx, Marcin Pilipczuk, and Michal Pilipczuk.
\newblock Hitting forbidden subgraphs in graphs of bounded treewidth.
\newblock {\em Inf. Comput.}, 256:62--82, 2017.

\bibitem[CMR00]{CourcelleMR00}
Bruno Courcelle, Johann~A. Makowsky, and Udi Rotics.
\newblock Linear time solvable optimization problems on graphs of bounded
  clique-width.
\newblock {\em Theory Comput. Syst.}, 33(2):125--150, 2000.

\bibitem[CNP{\etalchar{+}}11]{CyganNPPRW11}
Marek Cygan, Jesper Nederlof, Marcin Pilipczuk, Michal Pilipczuk, Johan M.~M.
  van Rooij, and Jakub~Onufry Wojtaszczyk.
\newblock Solving connectivity problems parameterized by treewidth in single
  exponential time.
\newblock In Rafail Ostrovsky, editor, {\em {IEEE} 52nd Annual Symposium on
  Foundations of Computer Science, {FOCS} 2011, Palm Springs, CA, USA, October
  22-25, 2011}, pages 150--159. {IEEE} Computer Society, 2011.

\bibitem[Cou90]{Courcelle90}
Bruno Courcelle.
\newblock The monadic second-order logic of graphs. {I}. {R}ecognizable sets of
  finite graphs.
\newblock {\em Inf. Comput.}, 85(1):12--75, 1990.

\bibitem[FGL{\etalchar{+}}19]{FominGLSZ19}
Fedor~V. Fomin, Petr~A. Golovach, Daniel Lokshtanov, Saket Saurabh, and Meirav
  Zehavi.
\newblock Clique-width {III:} hamiltonian cycle and the odd case of graph
  coloring.
\newblock {\em {ACM} Trans. Algorithms}, 15(1):9:1--9:27, 2019.

\bibitem[FGLS10]{FominGLS10}
Fedor~V. Fomin, Petr~A. Golovach, Daniel Lokshtanov, and Saket Saurabh.
\newblock Intractability of clique-width parameterizations.
\newblock {\em {SIAM} J. Comput.}, 39(5):1941--1956, 2010.

\bibitem[FGLS14]{FominGLS14}
Fedor~V. Fomin, Petr~A. Golovach, Daniel Lokshtanov, and Saket Saurabh.
\newblock Almost optimal lower bounds for problems parameterized by
  clique-width.
\newblock {\em {SIAM} J. Comput.}, 43(5):1541--1563, 2014.

\bibitem[FHRV08]{FioriniHRV08}
Samuel Fiorini, Nadia Hardy, Bruce~A. Reed, and Adrian Vetta.
\newblock Planar graph bipartization in linear time.
\newblock {\em Discret. Appl. Math.}, 156(7):1175--1180, 2008.

\bibitem[Klo94]{kloks1994treewidth}
Ton Kloks.
\newblock {\em Treewidth: computations and approximations}, volume 842.
\newblock Springer Science \& Business Media, 1994.

\bibitem[LMS11a]{lokshtanov2011known}
Daniel Lokshtanov, D{\'a}niel Marx, and Saket Saurabh.
\newblock Known algorithms on graphs of bounded treewidth are probably optimal.
\newblock In {\em Proceedings of the Twenty-Second Annual ACM-SIAM Symposium on
  Discrete Algorithms}, pages 777--789. SIAM, 2011.

\bibitem[LMS11b]{LokshtanovMS11soda}
Daniel Lokshtanov, D{\'{a}}niel Marx, and Saket Saurabh.
\newblock Slightly superexponential parameterized problems.
\newblock In Dana Randall, editor, {\em Proceedings of the Twenty-Second Annual
  {ACM-SIAM} Symposium on Discrete Algorithms, {SODA} 2011, San Francisco,
  California, USA, January 23-25, 2011}, pages 760--776. {SIAM}, 2011.

\bibitem[LMS18a]{LokshtanovMS18known}
Daniel Lokshtanov, D{\'{a}}niel Marx, and Saket Saurabh.
\newblock Known algorithms on graphs of bounded treewidth are probably optimal.
\newblock {\em {ACM} Trans. Algorithms}, 14(2):13:1--13:30, 2018.

\bibitem[LMS18b]{LokshtanovMS18slightly}
Daniel Lokshtanov, D{\'{a}}niel Marx, and Saket Saurabh.
\newblock Slightly superexponential parameterized problems.
\newblock {\em {SIAM} J. Comput.}, 47(3):675--702, 2018.

\bibitem[MRRS12]{misra2012parameterized}
Pranabendu Misra, Venkatesh Raman, MS~Ramanujan, and Saket Saurabh.
\newblock Parameterized algorithms for even cycle transversal.
\newblock In {\em International Workshop on Graph-Theoretic Concepts in
  Computer Science}, pages 172--183. Springer, 2012.

\bibitem[Pil11]{DBLP:conf/mfcs/Pilipczuk11}
Michal Pilipczuk.
\newblock Problems parameterized by treewidth tractable in single exponential
  time: {A} logical approach.
\newblock In Filip Murlak and Piotr Sankowski, editors, {\em Mathematical
  Foundations of Computer Science 2011 - 36th International Symposium, {MFCS}
  2011, Warsaw, Poland, August 22-26, 2011. Proceedings}, volume 6907 of {\em
  Lecture Notes in Computer Science}, pages 520--531. Springer, 2011.

\bibitem[RS84]{RobertsonS84}
Neil Robertson and Paul~D. Seymour.
\newblock Graph minors. {III.} planar tree-width.
\newblock {\em J. Comb. Theory, Ser. {B}}, 36(1):49--64, 1984.

\bibitem[SdSS20]{DBLP:conf/mfcs/SauS20}
Ignasi Sau and U{\'{e}}verton dos Santos~Souza.
\newblock Hitting forbidden induced subgraphs on bounded treewidth graphs.
\newblock In Javier Esparza and Daniel Kr{\'{a}}l', editors, {\em 45th
  International Symposium on Mathematical Foundations of Computer Science,
  {MFCS} 2020, August 24-28, 2020, Prague, Czech Republic}, volume 170 of {\em
  LIPIcs}, pages 82:1--82:15. Schloss Dagstuhl - Leibniz-Zentrum f{\"{u}}r
  Informatik, 2020.

\bibitem[Wan94]{Wanke94}
Egon Wanke.
\newblock k-nlc graphs and polynomial algorithms.
\newblock {\em Discret. Appl. Math.}, 54(2-3):251--266, 1994.

\end{thebibliography}

\end{document}